%% file: main.tex
\documentclass{article}

\usepackage[english]{babel}
\usepackage{tikz}
\usetikzlibrary{automata, positioning, arrows}

\usepackage[letterpaper,top=1in,bottom=1in,left=1in,right=1in]{geometry}

\usepackage{amsmath}
\usepackage{graphicx}
\usepackage{amsthm}
\usepackage[colorlinks=true, allcolors=blue]{hyperref}
\usepackage{color}
\usepackage[boxed,vlined,ruled,linesnumbered]{algorithm2e}
\usepackage{subcaption}
\usepackage{multirow}
\usepackage{colortbl}

\newtheorem{lemma}{Lemma}

\newtheorem{definition}{Definition}
\newtheorem{corollary}{Corollary}
\newtheorem{fact}{Fact}

\newtheorem*{thmA*}{Main Algorithmic Theorem}
\newtheorem*{thmL*}{Lower Bound Theorem}

\usepackage{thmtools, thm-restate}
\declaretheorem{theorem}

\newcommand{\dk}[1]{{\color{purple} #1}}
\newcommand{\pga}[1]{{\color{red} #1}}
\newcommand{\mm}[1]{{\color{cyan}{#1}}}
\newcommand{\mam}[1]{{\color{cyan}{#1}}}
\newcommand{\mamr}[1]{{\color{cyan}{#1}}}
\newcommand{\todo}[1]{{\color{red} [TODO: #1]}}
\renewcommand{\dk}[1]{#1}
\renewcommand{\mm}[1]{#1}
\renewcommand{\mam}[1]{#1}
\renewcommand{\mamr}[1]{#1}
\renewcommand{\pga}[1]{#1}

\newcommand{\mF}{\mathcal{F}}
\newcommand{\cP}{\mathcal{P}}

\newcommand{\cM}{\mathcal{M}}
\newcommand{\cN}{\mathcal{N}}
\newcommand{\remove}[1]{}
\newcommand{\polylog}{\text{\em\ polylog }}%{\mbox{ polylog } }

\usepackage{xspace}
\newcommand{\parhead}[1]{\noindent{\textbf{#1.}\xspace}}
\newcommand{\congest}{{\fontfamily{cmss}\selectfont CONGEST}\xspace}
\newcommand{\alg}{{\sc c2b}\xspace}

\title{Beeping Deterministic \congest Algorithms in Graphs
}
\author{
Pawel Garncarek\thanks{University of Wroclaw, Institute of Computer Science, Wroclaw, Poland; supported by the National Science Center, Poland (NCN), grant 2020/39/B/ST6/03288.} 
\and Dariusz R. Kowalski\thanks{Augusta University, Department of Computer \& Cyber Sciences, Augusta, GA, USA} 
\and Shay Kutten\thanks{Technion, Israel Institute of Technology, Haifa, Israel; a large part of this author's research was performed while he was on a sabbatical at Fraunhofer SIT in Darmstadt. Research supported in part by the Israeli Science Foundation and by The Bernard M. Gordon Center for Systems Engineering at the Technion.} 
\and Miguel A. Mosteiro\thanks{Pace University, Computer Science Department, New York, NY, USA; partially supported by Pace SRC grant and Kenan fund.}
}

\date{}

\begin{document}
\maketitle

%\vspace*{-5.1ex}

\begin{abstract}
\input{abstract}
\end{abstract}

% to be removed
%\input{notes}

\thispagestyle{empty}

\newpage

 \setcounter{page}{1}

\input{intro}

%\input{table}
\input{related}

\input{model}

\input{primitives_intro}

%%%%%%%%%%%%%%%%%%%%%%%%%%%%%%%%%%%%%%%%%%%%%%%%%%%%%%%%%%%

\vspace*{-3ex}
\section{Simulation of a \congest Round in Beeping Networks}
\label{sec:main-simulation}

%We already showed how to efficiently simulate a single round of a \congest model in a general beeping network, provided each node wants to send the same message to all of its neighbors. It already allows us to simulate many algorithms designed for the \congest model on a beeping network, in a deterministic and distributed way. 
Unfortunately, not all efficient graph algorithms in the \congest networks have the property of always broadcasting the same (short) message to every neighbor, which we exploit in Section~\ref{sec:local-broadcast}.\footnote{%
Note that in the LOCAL model, where the sizes of messages are of second importance (as long as they are polynomial), nodes can combine individual messages into one joint message and send it to all neighbors.}
In this section, we present a deterministic distributed algorithm that simulates a round of {\em any algorithm in the \congest model}, even if the algorithm sends different messages to neighbors.
It is only somewhat (polylogarithmically) slower than the more restricted one
(local broadcast, which required a node to send the same message to all its neighbors),
given in Theorem~\ref{th:local_broadcast} and Section~\ref{sec:local-broadcast}, but it is adaptive and uses heavier machinery.
%-- we show later that any non-adaptive solutions (beeping codes) are substantially less efficient and require $\Omega(\Delta^3)$ rounds. \mm{[[MM: do we prove this??]]}
 % we had to construct an adaptive algorithm using a technique for more complex codes
%  (puting a hifor the general simulator.
  %
 The novel construction is built hierarchically using the known family of code called ``avoiding selectors''. This, intuitively, already says ``when to beep.'' However, it is still possible that, for example, when two neighbors of   some node $v$ 
 %each 
 send a (different) message, of multiple bits per message, node $v$ will receive a ``message'' that is a logical OR of the two.
This is efficiently resolved by the new adaptive algorithm by employing a 3-stage handshake procedure, which sends pieces of the code that now serve in identifying what the IDs and messages are; it allows to spot overlapping transmissions from more than one neighbor and successfully decode those that do not overlap. Intuitively, each stage is ``triggered'' by a different level of the code.

%taken for different parameters. 
% 3-stage adaptive  
%handshake
%procedure (announcing, responding and confirming). on the top of the code. 
%which is an adaptive part. 
%This procedure is tightly correlated with the hierarchical structure of the code -- higher level code triggers announcing, while the lower level codes trigger responding and confirming.}

\vspace*{-2ex}
\paragraph{Preliminaries and Challenges.}
Suppose every node has a possibly different message to deliver to each of its neighbors. We could use the algorithm from Section~\ref{sub:neighbourhood} to learn neighbors' IDs first in $O(\Delta^2\log^2 n)$ beeping rounds.
W.l.o.g., assume that each message from a node $v$ to a node $w$ has $O(\log n)$ bits
%; otherwise, we could easily split it into chunks of size $O(\log n)$ and apply our new algorithm to each of them, sequentially 
(otherwise, 
%asymptotic formula on 
the bound on the time complexity is increased by a factor of $\cM/\log n$, where $\cM$ is the maximum size of a single message).
Simulation of a \congest round faces the following challenges.

\noindent
{\em Challenge 1.}
A node could try to compute its beeping schedule to avoid overlapping with other neighbors of the receiving node. However, it requires knowing at least $2$-hop neighborhood, which is costly. 
%(even $1$-hop requires $\Omega(\Delta^2\log n)$ rounds).
%The first main challenge to overcome is that a node $v$ does not know its $2$-hop neighborhood graph. Learning it could have potentially helped $v$ to use the specific differences between the identities to decide on a schedule when a successful transmission to/from each specific neighbor could take place (especially if transmission schedules are fixed, and nodes only decide what message to beep); however, it would require retrieving up to $\Theta(\Delta^2)$ node IDs. 
A node could try to learn first the IDs of its $1$-hop neighbors, and then broadcast them, one after another, using the local broadcast algorithm, 
%from Section~\ref{sec:local-broadcast}, 
but since there could be $\Theta(\Delta)$ such IDs (each represented by $O(\log n)$ bits), the overall time complexity would be $O(\Delta^3 \log^3 n)$, by Theorem~\ref{th:local_broadcast}.
Instead, our algorithm uses specific codes, called avoiding selectors (see Definition~\ref{def:avoid-selector}), to assure partial progress in information exchange in periods that sum up to $\Theta(\Delta^2 \polylog n \log \Delta)$.

\noindent
{\em Challenge 2.}
%Even knowing its $1$-hop neighbors, 
A node has to choose which of its input messages to beep at a time or find a more complex beeping code to encode many of its input messages. 
If it chooses ``wrongly,'' the message could be ``jammed'' by other beeping neighbors of the potential receiver.
To overcome this, avoiding selectors ensures that many nodes ``announce'' themselves successfully (i.e., without interference) to many of their neighbors, and these ``responders'' use avoiding selectors to respond. Once an announcer hears the ID of its responder, the handshaking procedure allows them to fix rounds for their point-to-point, non-interrupted communication.
%Such a chosen message could be arbitrarily ``jammed'' by some $2$-hop neighbors, which are not initially known to the node (see Challenge 1 above). 
%We prove later in Section~\ref{} \textcolor{red}{---------PENDING} that this is indeed a challenge, and any {\em non-adaptive beeping code} requires $\Omega(\Delta^3 \log\Delta n)$ \mm{[[MM: this might be confusing, is it n times log ?? ]]} rounds.
%To overcome it, the avoiding selectors (mentioned above) could be used in an adaptive way with properly chosen parameters to guarantee initiations of beeping communication in many links that are ``isolated'' in the network. We exploit them by designing a system of hand-shaking procedures, organized in three types of longer messages to beep: announcing (a node, called an announcer, beeps that it wants to communicate), responding (a node retrieving a message from an announcer, beeps a message destined to the announcer) and confirming (the announcer beeps a confirmation message destined to the responder). If indeed such links are isolated enough in the graph, we prove that this process guarantees successful message exchange between the announcer and the responder.

%the ad hoc topology and, consequently, 

The abovementioned avoiding selectors for $n$ nodes are parameterized by two numbers, $k,\ell$, corresponding to the number of competing neighbors/responders versus the other (potentially interrupting) neighbors:
%their definition~follows:

\begin{definition}[Avoiding selectors]
\label{def:avoid-selector}
    A family $\mathcal{F}$ of subsets of $[n]$ 
    %of size at most $k$ each 
    is called an \emph{$(n,k,\ell)$-avoiding selector}, where $1\le \ell < k\le n$, if for every non-empty subset $S \in [n]$ such that $|S| \leq k$ and for any subset $R\subseteq S$ of size at most $\ell$, there is an element $a \in S\setminus R$ for which there exists a set $F \in \mathcal{F}$ such that $|F \cap S| = \{a\}$.
\end{definition}

The following fact follows directly from Definition~\ref{def:avoid-selector}, see also \cite{BonisGV05,ChlebusK05}.

\begin{fact}
\label{fact:avoiding-selectors}
Suppose we are given an $(n,k,\ell)$-avoiding selector $\mathcal{F}$ and a set $S$ of size at most $k$.
Then, the number of elements in $S$ not ``selected'' by selector $\mF$ (i.e., for which there is no set in the selector that intersects $S$ on such singleton element) is smaller than $k-\ell$.
\end{fact}

\begin{theorem}[\cite{BonisGV05,ChlebusK05}]
\label{thm:avoiding-selectors}
There exists an $(n,k,\ell)$-avoiding selector of length $O\left(\frac{k^2}{k-\ell}\log n\right)$, and moreover, an $(n,k,\ell)$-avoiding selector of length $O\left(\frac{k^2}{k-\ell}\text{\em\ polylog } n\right)$ can be efficiently deterministically constructed (in polynomial time of $n$) for some polylogarithmic function $\text{\em\ polylog } n$, locally by each~node.
\end{theorem}

%\subsection{Main deterministic distributed algorithm simulating any \congest~round}
\subsection{The \alg Algorithm}
\label{sec:main-general-algorithm}

%The main 
Our simulator algorithm proceeds in epochs $i=1,\ldots\log\Delta $. 
A pseudo-code for an epoch $i$ is provided at the end of this subsection.
In the beginning, each node has all its links not successfully realized -- here by a link $\{v,w\}$ being realized we understand that up to the current round, an input message/ID sent by $v$ (using a sequence of beeps) has been successfully encoded by $w$ and vice versa (note that these are two different messages and were sent/encoded each in a different round); the formal definition of link realization will be given later.
The goal of the algorithm is to preserve the following invariant for epoch $i\ge 1$: 
\begin{quote}
\hspace*{-1em}
At the end of epoch 
$i=1,\ldots,\log \Delta$, 
%At the beginning of epoch 
%$i=1,\ldots,\log_{3/2} \Delta$, 
each vertex has less than 
$\kappa_i= \Delta / 2^{i}$ 
%$\kappa_i= \Delta \cdot (2/3)^{i-1}$ 
incident links not realized. 
\end{quote}
We also set an auxiliary value $\kappa_0=\Delta$, which corresponds to the maximum number of adjacent links per node at the beginning of the computation. For ease of presentation, we assume that node IDs come from the range $[1,n]$. Note that in all the formulas, the number of possible IDs appears only under logarithms, so the algorithm and proof for range $[1,n^c]$ are the same.
%By definition, $\kappa_1\le \Delta$.

\vspace*{-1.5ex}
%\paragraph{Main algorithm for epoch $i$.}
\paragraph{Algorithm for epoch $i$: Preliminaries and main concepts.} 

Epoch $i$ proceeds in subsequent batches of $2\log n$ rounds, each batch is called a \defn{super-round}. In a single super-round, a node can constantly listen or keep beeping according to some 0-1 sequence of length $2\log n$, where 1 corresponds to beeping in the related round and 0 means staying silent. 
The sequences that the nodes use during the algorithm are \defn{extended-IDs}, defined as follows: the first $\log n$ positions contain an ID of some node in $\{1,\ldots,n\}$, while the next $\log n$ positions contain the same ID 
%but 
\mamr{with the bits flipped, that is,}
with ones swapped to zeros and vice versa. 
Note that extended-IDs are pairwise different, and each of them contains exactly $\log n$ ones and $\log n$ zeros.
We say that a node $v$ {\em beeps an extended-ID of node $w$ in a super-round} 
%if $v$ keeps beeping exactly in rounds corresponding to the extended-ID of $w$ within this super-round 
\mamr{$s$ if, within super-round~$s$, node $v$ beeps only in rounds corresponding to positions with $1$'s in the extended ID of $w$ } 
($w$ could be a different node id than $v$).
We say that a node $w$ {\em receives an extended-ID of a node $v$ in a super-round} \mamr{$s$} if:
\vspace*{-0.8ex}
\begin{itemize}
\item 
$w$ does not beep in super-round \mamr{$s$},
\vspace*{-0.8ex}
\item 
the sequence of 
%beeps received 
\mamr{noise/silence heard by $w$}
in super-round \mamr{$s$} form an extended-ID of $v$.
\end{itemize}

\vspace*{-0.8ex}
\noindent
From the perspective of receiving information in a super-round, all other cases not falling under the above definition of receiving an extended-ID, i.e., when a node is not silent in the super-round or receives a sequence of beeps that does not form any extended-ID, are ignored by the algorithm, in the sense that it could be treated as meaningless information noise. 

%Analogously to extended-ID, each node $w$ creates an {\em extended-message addressed to a neighbor $v$}. Node $w$ does it 
\mamr{Analogously to extended-ID's, nodes create an \defn{extended-message}} by taking the binary representation of the message of logarithmic length and transforming it to a $2\log n$ binary sequence in the same way as an extended-ID is created from the binary ID of a node.
An extended-message, as well as an extended-ID, is easily decodable after being received without interruptions from other neighbors.

%A crucial definition specifies what does it mean to one-to-one communications, as given for free in the \congest model, in beeping networks. 
\mamr{A specification of the conditions to achieve one-to-one communication, which is given ``for free'' in the \congest model, is crucial. An illustration of the following handshake communication procedure is shown in Figure~\ref{fig:alg}.}
We say that our algorithm \defn{realizes link $\{v,w\}$} if the following 
%conditions
are~satisfied:
\vspace*{-0.8ex}
\begin{itemize}
\item[(a)] 
there are three consecutive super-rounds (called ``responding'') in which $v$ beeps an extended-ID of itself followed by an extended-ID of $w$ and then by extended-message of $v$ addressed to $w$, and $w$ receives them in these super-rounds; intuitively, it corresponds to the situation when $v$ ``tells'' $w$ that it dedicates these three super-rounds for communication from itself to $w$, and $w$ receives this information;
\vspace*{-3.5ex}
\item[(b)] 
there are three consecutive super-rounds (called ``confirming'') in which $w$ beeps an extended-ID of itself followed by an extended-ID of $v$ and by its extended-message addressed to $v$, and $v$ receives them in these super-rounds; intuitively, it corresponds to the situation when $w$ ``tells'' $v$ that it dedicates these three super-rounds for communication from itself to $v$, and $v$ receives this information;
\vspace*{-0.8ex}
\item[(c)]
there is a super-round, not earlier than the one specified in point (a), at the end of which node $w$ locally marks link $\{v,w\}$ as realized, 
and analogously, 
there is a super-round, not earlier than the one specified in point (b), at the end of which node $v$ locally marks link $\{v,w\}$ as realized.
\end{itemize}

\vspace*{-0.8ex}
\noindent
It is straightforward to see that in super-rounds specified in points (a) and (b), a multi-directional communication between $v$ and $w$ takes place -- by sending and receiving both ``directed pairs'' of extended-IDs of these two nodes, each of them commits that the super-rounds specified in points (a) and (b) are dedicated for sending a message dedicated to the other node, and vice versa. Additionally, in some super-round(s) both nodes commit that it has happened (c.f., point (c) above).

\vspace*{-2ex}
\paragraph{Algorithm for epoch $i$: Structure.} 

\input{algorithm}

An epoch $i$ is split into $|\mF_{\Delta,k_i}|$ {\em phases}, for a given $(n,\Delta,\Delta - k_i)$-avoiding selector $\mF_{\Delta,k_i}$ and parameter $k_i=\Delta/2^i$, parameterized by a variable $j$. Each phase starts with one {\em announcing super-round}, in which nodes in set $\mF_{\Delta,k_i}(j)$ beep in pursuit to be received by some of their neighbors. This super-round is followed by $\log k_i$ {\em sub-phases}, parameterized by $a=1,\ldots, \log k_i$. A sub-phase $a$ uses sets from an $(n,k_i/2^{a-2},k_i/2^{a-1})$-avoiding selector $\mF_{k_i/2^{a-2},k_i/2^{a-1}}$ to determine who beeps in which super-round (together with additional rules to decide what extended-ID and extended-message to beep and how to confirm receiving them), and consists of 
%$\sum_{a=1}^{\log k_i} 
$|\mF_{k_i/2^{a-2},k_i/2^{a-1}}|$ %quadruples 
$6$-tuples
of super-rounds ($3$ responding super-rounds and $3$ confirming super-rounds). 
The goal of a phase is to realize links that were successfully received (``announced'') in the first (announcing) super-round of this phase. This is particularly challenging in a distributed setting since many neighbors could receive such an announcement, but the links between them and the announcing node must be confirmed so that one-to-one communication between the announcer and responders could take place in different super-rounds (in one super-round, a node can receive only logarithmic-size information).

\input{algannouncer}

\mamr{
%\paragraph{Pseudo-code for epoch $i$.} 
%Below is a detailed description of the algorithm for {\bf\em Epoch $i$}.
\vspace*{-2ex}
\paragraph{Algorithm for epoch $i$: Definitions and notation.} 
The pseudocode of the \alg algorithm can be seen in Algorithm~\ref{algC2Bv2}, and its subroutines in Algorithms~\ref{algC2Bv2A} and~\ref{algC2Bv2L}.
$\mF_{\Delta,k_i}$ is a locally computed $(n,\Delta,\Delta-k_i)$-avoiding selector, and for any $a=1,\ldots,\log k_i$,
%selector 
$\mF_{k_i/2^{a-2},k_i/2^{a-1}}$ is a (locally computed) $(n,k_i/2^{a-2},k_i/2^{a-1})$-avoiding selector, as in Theorem~\ref{thm:avoiding-selectors}.
We denote the extended-ID of node $x$ as $\langle x\rangle$, and the extended-message of node $x$ for node $y$ as $\langle m_{x,y}\rangle$, both given as a sequence of bits. For any sequence of bits $s$, $s(i)$ is the $i^{th}$ bit of $s$.
}

\vspace*{-2ex}
\subsection{Analysis of the \alg Algorithm}

Recall that the algorithm proceeds in synchronized super-rounds, each containing a subsequent $2\log n$ rounds. Therefore, our analysis assumes that the computation is partitioned into consecutive super-rounds and, unless stated otherwise, it focuses on correctness and progress in super-rounds. 
Recall also that each node either stays silent (no beeping at all) or beeps an extended ID of some node or an extended message of one node addressed to one of its neighbors in a super-round.
The missing proofs 
%from this section 
are deferred to Section~\ref{sec:proofs-main-simulation}.

In the next two technical results, we state and prove the facts that receiving an extended-ID by a node $w$ in a super-round can happen if and only if there is {\em exactly one neighbor} of $w$ has been beeping {\em the same extended-ID} during the considered super-round. 

\begin{fact}[Single beeping]
\label{fact:single-beeping}
If during a super-round, exactly one neighbor of a node $w$ beeps an extended-ID of some $z$, then $w$ receives this extended-ID in this super-round.
\end{fact}

\begin{proof}
Directly from the definition of receiving an extended-ID. 
\end{proof}

\begin{lemma}[Correct receiving]
\label{lem:correct-receiving}
During the algorithm, if a node $w$ 
%stays silent and 
receives some extended-ID of $z$ in a super-round, then some unique neighbor $v$ of $w$ has been beeping an extended-ID of $z$ in this super-round while all other neighbors of $w$ have been silent. 
%in this super-round. 
The above holds except, possibly, some second responding super-rounds, in which a node can receive an extended-ID of $z$ that has been beeped by more than one neighbor.
\end{lemma}

%\begin{minipage}{1\linewidth}
%\input{algannouncer}
\input{alglistener}

%\end{minipage}

We now prove that link realization implemented by our algorithm is consistent with the definition -- it allocates in a distributed way super-rounds for bi-directional communication of distinct messages.

\begin{lemma}[Correct realization]
\label{lem:correct-realization}
If a node $v$ (locally) marks some link $\{v,w\}$ as realized, which may happen only at the end of a second confirming super-round, the link has been realized by then. 
\end{lemma}

As mentioned earlier in the description of the phase, the goal of a phase $j$ (of epoch $i$) is to assure that any node $v$ that was received by some other nodes $w$ in the announcing super-round, gets all such links $\{v,w\}$ realized by the end of the phase (and vice versa, because the condition on the realization by this algorithm is symmetric).
The next step is conditional progress in a sub-phase $a$ of a phase $j$.

\begin{lemma}[Sub-phase progress]
\label{lem:subphase-progress}
Consider any node $v$ and suppose that in the beginning of sub-phase $a$ of phase $j$, there are at most $\Delta/2^{i+a-2}$ nodes $w$ such that $w$ is $(j,v)$-responsive and it does not mark link $\{v,w\}$ as realized. Then, by the end of the sub-phase, the number of such nodes is reduced to less~than~$\Delta/2^{i+a-1}$.
\end{lemma}

\begin{lemma}[Phase progress]
\label{lem:phase-progress}
Consider a phase $j$ of epoch $i$ and assume that in the beginning, there are at most $2k_i$ non-realized incident links to any node. Every node $w$ that becomes $(j,v)$-responsive in the first (announcing) super-round of the phase, for some $v$, mark locally the link $\{v,w\}$ as realized during this phase. And vice versa, also node $v$ marks locally that link as realized. 
\end{lemma}

The next lemma proves the invariant for epoch $i$, assuming that it holds in the previous epochs. 

\begin{lemma}[Epoch invariant]
\label{lem:epoch-invariant}
The invariant for epoch $i\ge 1$ holds. 
\end{lemma}

\begin{theorem}
\label{thm:congest-sim}
%The main deterministic distributed algorithm 
The \alg algorithm deterministically and distributedly
simulates any round of any algorithm designed for the \congest networks in $O(\Delta^2 \polylog n \log\Delta)$ beeping rounds, where the $\polylog n$ is the square of the (poly-)logarithm in the construction of avoiding-selectors in Theorem~\ref{thm:avoiding-selectors} multiplied by $\log n$.
\end{theorem}

%\sk{REPEATING THE THEOREM IN 2 PLACES  SEEMS UNJUSTIFIABLY COSTLY IN TERMS OF SPACE REAL ESTATE. I WOULD START DIRECTLY WITH "PROOF OF THEOREM 6". OR REMOVE IN THE INTRO. tHE SAME GOES WITH THE COROLLARIES GIVEN IN THE INTRO. ALREADY HALF A PAGE SAVING.  }

\begin{proof}
By Lemma~\ref{lem:epoch-invariant}, each epoch $i$ reduces by at least half the number of non-realized incident links. 
We next count the number of rounds in each epoch by counting the number of super-rounds and multiplying the result by the $O(\log n)$ length of each super-round.
Recall that link realization means that some triples of responding and confirming super rounds were not interrupted by other neighbors of both end nodes of that link; therefore, the attached extended messages (in the third super-rounds in a row) were correctly received. Thus, the local exchange of messages addressed to specific neighbors took place successfully.

Each sub-phase $a$ has $O(\Delta^2 \polylog n)$ super-rounds, because for each set in of the $(n,k_i/2^{a-1},k_i/2^a)$-avoiding selector $\mF_{k_i/2^{a-1},k_i/2^a}$, there are four super-rounds and the selector itself has $O((k_i/2^a) \polylog n)$ set, by Theorem~\ref{thm:avoiding-selectors}.

Therefore, the total number of super-rounds in all sub-phases executed 
within 
%point~\ref{alg:sub-phase} of the algorithm~
\mamr{the loops in Line~\ref{line:subphaseloopA} of Algorithm~\ref{algC2Bv2A} and Line~\ref{line:subphaseloopL} of Algorithm~\ref{algC2Bv2L}}
is 
\vspace*{-1.3ex}
\[
O(\sum_{a=1}^{\log k_i} (k_i/2^a) \polylog n) \le
O( k_i \polylog n)
\ .
\]

\vspace*{-0.7ex}
\noindent
Within one phase, they are executed as many times as the number of announcing super-rounds. 
The number of announcing super-rounds in a phase is
$|\mF_{\Delta,k_i}|$, which is $O((\Delta^2/k_i)\cdot \polylog n)$ by Theorem~\ref{thm:avoiding-selectors}.
Hence, the total number of super-rounds in a phase is 
%\[
$O( (\Delta^2/k_i)\cdot \polylog n \cdot k_i \polylog n)
\le 
O(\Delta^2 \polylog n)$,
%\ ,
%\]
where the final $\polylog n$ is a square of the (poly-)logarithms from Theorem~\ref{thm:avoiding-selectors}.

Since there are $\log\Delta$ epochs, the total number of super-rounds is $O(\Delta^2 \polylog n \log\Delta)$, which is additionally multiplied by $O(\log n)$ -- the length of each super-round -- if we want to refer the total number of beeping rounds.
%
%super-rounds of $O(\log n)$ beeping rounds each
\end{proof}

\vspace*{-2.5ex}
\paragraph{Maximal Independent Set (MIS):}
To demonstrate that the above efficient simulator can yield efficient results for many graph problems, we apply it to the 
algorithm of~\cite{ghaffari2021improved}% 
% \dk{??? and others such as ????}
 to improve polynomially (with respect to $\Delta$) the best-known solutions 
for MIS (c.f. \cite{beauquier2018fast}):

%\dk{(c.f.,~\cite{???})}:

%\todo{More uses for Network Decomposition! Add corollaries here and citations in Related Work.}

\begin{corollary}
\label{cor:mis}
% Graph problems such as MIS, \dk{????????????} 
MIS can be solved deterministically on any network of maximum node-degree $\Delta$ in $O(\Delta^2 \polylog n)$ beeping rounds.
\end{corollary}

\remove{%%%%%%%%%%%%%%%%%

\subsection{Cubic Lower Bound for Non-adaptive Beeping Codes}
\label{sec:lower-non-adaptive}

Consider the following simplification of the simulation problem. Each node $v$ is given, as an input, parameters $n,\Delta$ and a vector of numbers in $[n]$ of length $x_v\le \Delta$. The goal is: for any graph $G$ such that $x_v=|N(v)|$, for any node $v$, every $i$th neighbor of $v$ (according to the order of IDs) learns the $i$th number in the vector of $v$, for any $1\le i \le x_v$.
We call this problem {\em local ports' learning}, as we could think about the numbers in the input vectors as (arbitrary) labels of ports from the node to its corresponding neighbor.

\begin{theorem}
Any beeping code solving the local ports' learning problem has length $\Omega(\Delta^3\log n)$.
\end{theorem}

\begin{proof}

\end{proof}

}%%%%%%%%%%%  END  REMOVE  %%%%%%%%%%

%%%%%%%%%%%%%%%%%%%%%%%%%%%%%%%%%%%%%%%%%%%%%%%%%%%%%%%%%%%
\input{multihop}

%\input{localbroadcastlowerbound}
\input{simulator_proofs}

\input{primitives_details}

\input{conclude}

%\bibliographystyle{alpha}
\bibliographystyle{plain}
\bibliography{bibliography}

\newpage

\appendix

\end{document}

%% file: abstract.tex
The Beeping Network (BN) model captures important properties of 
biological processes, for instance, when the beeping entity, called \emph{node}, models a cell.
 %(see Navlakha and Bar-Josef, CACM 2014, and Afek et al., Science 2011).
 Perhaps paradoxically, even the fact that the communication capabilities of such nodes are extremely limited has helped 
BN become one of the fundamental models for networks where nodes' transmissions interfere with each other. Since in each round, a node may transmit at most one bit, 
it is useful to treat the communications in the network as distributed coding and design it to overcome the interference. We study both non-adaptive and adaptive codes. Some communication and graph problems already studied in the Beeping Networks admit fast (i.e., polylogarithmic in the network size $n$) \emph{randomized} algorithms.
On the other hand, all known \emph{deterministic} algorithms for non-trivial problems have time complexity (i.e., the number of beeping rounds, corresponding to the length of the used codes) at least polynomial in the maximum node-degree $\Delta$. 

We improve known results for deterministic algorithms by first showing that this polynomial can be as low as $\tilde{O}(\Delta^2)$. More precisely, we prove that beeping out a single round of any \congest algorithm in any network of maximum node-degree $\Delta$ can be done in $O(\Delta^2 \polylog n)$ beeping rounds, each accommodating at most one beep per node, even if the nodes intend to send different messages to different neighbors. This upper bound reduces polynomially the time for a \emph{deterministic} simulation of \congest in a Beeping network, comparing to the best known algorithms, and nearly matches the time obtained recently using \emph{randomization} (up to a poly-logarithmic factor). %(see Davies, ACM PODC 2023).
Our simulator allows us to implement any efficient algorithm designed for the \congest networks in the Beeping Networks, with $O(\Delta^2 \polylog n)$ overhead. This $O(\Delta^2 \polylog n)$ implementation results in a polynomial improvement upon the best-to-date $\Theta(\Delta^3)$-round beeping MIS algorithm (and of related tasks). Using a more specialized (and thus, more efficient) transformer and some additional machinery,  we constructed various other efficient deterministic Beeping algorithms for various other commonly used building blocks, such as 
 Network Decomposition (seminal in the field of \congest graph algorithms).
For $h$-hop simulations, we prove a lower bound $\Omega(\Delta^{h+1})$, and we design a nearly matching algorithm that is able to ``pipeline'' the information in a faster way than
 working layer by layer.
%We also prove that non-adaptive content-oblivious deterministic algorithms that use at least a single local broadcast require $\Omega(\min\{n,\Delta^2/\log^2 n\})$ beeping rounds in some networks of maximum node-degree $\Delta$. This lower bound establishes a gap between randomized and non-adaptive content-oblivious deterministic algorithms for many such tasks (in contrast to the general case where we have shown no gap exists, up to a polylog factor). 

\

\noindent
{\bf Keywords:} Beeping Networks, \congest Networks, deterministic simulations, 
%Local Broadcast, 
graph algorithms.

%Single-hop: leader election (Gilber and Newport, DISC 2012?) and weaker synchronization principles

%Multi-hop: MIS and distance $2$ MIS principles 

%Efficient translations of message-passing or population protocols or stone age model~\cite{???}

%% file: intro.tex
%%%%%%%%%%%%%%  END  REMOVE  %%%%%%%%%%

\section{Introduction}

The study of the Beeping Networks (BN) model is simultaneously interesting, challenging, and useful in several respects. 
Even 
%less restrictive ad-hoc multiple (e.g., wireless) networks 
multiple less-restrictive ad-hoc networks (e.g., wireless) 
are frequently the algorithmist delight: a seemingly simple computational problem that becomes challenging under harsh yet realistic conditions.  In terms of communication, Beeping Networks~\cite{cornejo2010deploying} may be the harshest model: network nodes can only \emph{beep} (emit a signal) or \emph{listen} (detect if a signal is emitted in its vicinity). A listening node may distinguish between \emph{silence} (no beeps) and \emph{noise} (at least one beep), but it cannot distinguish between a single beep and multiple. Theoretically, the Beeping model is important since it enables one to study whether distributed tasks can be performed efficiently under minimal conditions,
 \mam{and it is a fundamental model to study communication complexity on channels where signals are superimposed (by applying the OR operator), which makes it more challenging than the basic model in which transmissions on links are independent. Moreover, in a one-hop network,\footnote{\mam{For any given path of links connecting two nodes, the number of \emph{hops} is the number of links in such path.}} \dk{non-adaptive communication schedules in the Beeping model are} equivalent to superimposed codes, which have been widely applied not only to communication problems, but also to text alignments, bio-computing, data pooling, dimensionality reduction, and other~areas.}
%\sk{Beside the wireless networks motivation, Beeping Networks are also motivated by Biologically Inspired Distributed Algorithms; for example, a node may model a body cell, and a beep may be some molecule secreted by neighboring cells \cite{afek2013beeping,navlakhadistributed}.}

%Algorithms designed for Beeping Networks are useful for massive networks because they can be implemented on non-expensive devices with low energy consumption. For example, on Internet of Things (IoT) applications such as Sensor Networks where a massive number of nodes are deployed for monitoring purposes. 
%The model is also useful for studying natural communication in biological networks~\cite{navlakha2014distributed,afek2011biological} (e.g., fireflies or cells). \mm{Indeed, better understanding the power of biological communication processes may lead to better nature-inspired algorithms~\cite{MooreNeuron24}.} 
%The model 

In addition to its theoretical importance, 
this model is recognized as useful for studying natural communication in biological networks~\cite{navlakha2014distributed,afek2011biological} (e.g., 
%fireflies or 
cells). A better understanding of the power of biological communication processes may lead to better nature-inspired algorithms~\cite{MooreNeuron24}.  

It has been argued that the model is also a practical tool because algorithms designed for Beeping Networks can be implemented on non-expensive devices with low energy consumption. For example, nodes are deployed (in ad-hoc topologies) for monitoring in Internet of Things (IoT) applications such as Sensor Networks. It should be noted that a mechanism similar in properties to beeping, called ``busy tone,'' has been in wide use in wireless networks but for very limited purposes in channel access algorithms (as opposed to the Beeping model that is intended for general purposes distributed algorithms). See, e.g., \cite{tobagi1975packet,haas2002dual}.

%\mm{A restricted communication model such as Beeping Networks yields}
%However, such a restricted communication model makes 
%the task of designing algorithms very challenging. 
%\sk{I moved this comment of very challenging to the first paragraph. I think this is the main theme there.}

A wealth of successful research on computational problems in Beeping Networks has appeared in the literature (see, for instance \cite{beauquier2018fast,dufoulon2022beeping} and the references therein). Nevertheless, fundamental distributed computing questions 
%in distributed computing 
remain open in the context of Beeping Networks, especially for deterministic algorithms. 
On the other hand, the \congest ~\cite{peleg2000distributed} model has been profusely studied.
%
%The \congest model is a synchronous, message-passing model of communication where, in each round, nodes are limited to send a (possibly different) message of logarithmic size to each neighboring node.
%
A natural question that follows is how to efficiently transform \congest Network algorithms into Beepping Network algorithms.

%and Maximal Independent Set. 
%In the following sub-section, we detail our positive and negative results. 

\vspace*{-1ex}
\subsection{Our Contributions}

%\todo{We refer to "beeping rounds" and "\congest rounds", but these models were barely mentioned before. Is this fine? SOLUTION: Write a short paragraph about models in the introduction before Our Contributions.}

%\todo{Should we present the results in a more concise form? Such as a table? Or a summary of results in bullet points and the end of the section or at the beginning?}

%\todo{Redact this and add before the first paragraph: "We develop tools: Learning neighborhood, .... Then, we use those tools to obtain one of our main results, beeping simulation of \congest round."}

% \todo{Should we switch the order between Our Contributions and Related Work?}
\vspace*{-0.2ex}
The results are detailed and compared to previous work in Table \ref{table:relwork}.
Our main contributions are the near optimal\footnote{Optimal up to a polylogarithmic factor.} deterministic implementations of two simulators.
The main simulator efficiently translates any \congest algorithms \mam{(deterministic or randomized)} to the Beeping model.
Each \congest round is simulated using $O(\Delta ^2 \polylog n \log \Delta )$ rounds, improving the previous result of $O(\Delta^4 \log n)$ of \cite{beauquier2018fast}.
The other simulator translates (more efficiently -
$O(B\Delta^2 \log n)$ for sending $B$ bits)
the more specialized
\emph{local broadcast} operation 
(called ``broadcast \congest'' in \cite{davies2023optimal}) 
%\mm{[[MM: should we call it instead Broadcast \congest simulation as in Davies paper? Either way, we have to point out that simulating Broadcast CONGEST is the same as \emph{concurrent} local broadcast, which is what our algorithm does, or otherwise the comparison with Davies upper bound may be confusing]]} 
where, whenever a node sends a message to its neighbors, it sends all of them the same message. It seems interesting that for the specialized simulator, we managed to employ non-adaptive ``Beeping codes'', 
i.e.,
%that is, 
not needing to know the topology nor to adapt an action based on beeps heard from the channel; \dk{as such, its performance almost matches a lower bound $\Omega(\Delta^2\log_\Delta n)$ for {\em non-adaptive} algorithms~\cite{CLEMENTI2003337}.} 
%(It, in \mam{[[something's wrong here, anyway, do we need this parenthesis? if so, we should refer to sections for the other results as well]]} Sections \ref{sec:primitives} with the details in Section  \ref{sec:prim_details}.)
In contrast, for the more general simulator, we had to develop a technique that may be of interest in itself. We started with more complex codes and added a hierarchical structure. On top of that, we developed an adaptive algorithm that employs a novel three-stage handshake to distinguish between a genuine message and a message the adversary may have composed by colliding~several~transmissions.
  
 % we had to construct an adaptive algorithm using a techniq more complex codes
%  (puting a hifor the general simulator.
  %
% The construction is built hierarchically using the known family of so-called ``avoiding selectors''. On top of the above (intuitively, saying ``when to beep'',) we also add a 
%taken for different parameters. 
 %3-stage adaptive %acknowledgement 
%handshake
%procedure (announcing, responding and confirming) on the top of the code. 
%which is an adaptive part. 
%This procedure is tightly correlated with the hierarchical structure of the code -- higher level code triggers announcing, while the lower level codes trigger responding and confirming.}
We demonstrate the usefulness of the general simulator by improving the known results for the problem most heavily investigated under the Beeping model, namely, the Maximal Independent Set (MIS). We also demonstrate the use of the more specialized simulator (plus some additional machinery) to construct efficient algorithms for other common building blocks of algorithms, such as a node learning its neighborhood (in $O(\Delta^2 \log^2 n)$ beeping rounds), 
learning the whole cluster  ($O(\Delta^2 \log^4 n)$ for clusters of depth $O(\log ^2 n)$) for all the clusters in parallel, 
and Network Decomposition ($O(\Delta^2 \log^8 n)$).

Another set of results concerns whether a gap exists between the complexity of randomized and deterministic algorithms. For the task of a general simulation of \congest algorithms, the time taken by our \emph{deterministic} simulator 
%is 
%
($O(\Delta ^2 \polylog n \log \Delta )$),  
matches the known lower bound ($\Omega(\Delta^2 \log n)$~\cite{davies2023optimal}) up to a poly-logarithmic factor. 
The proof of the $\Omega(\Delta^2 \log n)$ lower bound in~\cite{davies2023optimal} is based on requiring each node to transmit a different string of bits to each of its neighbors.\footnote{This problem is 
%in fact 
named local broadcast 
%for the lower bound 
in~\cite{davies2023optimal}.
However, 
if a different message is sent to each neighbor, the reason for calling it broadcast is unclear. 
We reserve the name local broadcast for the classic problem of sending the {\em same message to~all~neighbors}.} If the string of bits could be the same, that is, the much simpler problem known as local broadcast, that lower bound would collapse to $\Omega(\Delta \log n)$.
\dk{Both lower bounds also hold for randomized solutions, which implies no (substantial) gap between randomness and determinism in the case of general \congest simulation algorithms.}
\remove{??? On the other hand, for local broadcast, we show a lower bound of $\Omega\left(\min\left\{n,\Delta^2/\log^2 n\right\}\right)$ for deterministic \mam{non-adaptive content-oblivious} beeping algorithms, which to the best of our knowledge is the first (nearly) quadratic \mam{local broadcast (same message for all neighbors)} lower bound in the Beeping model.
The complexity of randomized local broadcast is much better - $O(\Delta \log n)$~\cite{davies2023optimal}, establishing a gap for algorithms that use only local broadcasts. (Note that local broadcasts are a common primitive in various models). ???}

\input{table}

\mam{
An interesting generalization is to extend the notion of neighborhood to $h\geq 1$ hops. For the case where each node has a possibly different $B$-bit message to deliver to all nodes in its $h$-hop neighborhood, a problem that we call \emph{$B$-bit $h$-hop simulation}, we show bounds $O(h\cdot B\Delta^{h+2} \polylog n)$ and $\Omega(B\Delta^{h+1})$, 
%even using randomization. 
where the lower bound holds also for randomized solutions. Our algorithm, on the other hand, efficiently ``pipelines'' point-to-point messages, and achieves substantially better complexity (for $h>3$) than a straightforward application of 1-hop simulator $h$ times (which would need $O(\Delta^{2h}\polylog n)$ time).
For the simpler problem when each node has to deliver the same message to all nodes in its $h$-hop neighborhood, which we call \emph{$B$-bit $h$-hop Local Broadcast}, we show a lower bound of $\Omega(B\Delta^h)$ even with randomization. 
}

%% file: table.tex
\begin{table}[t]
    \vspace*{-1ex}
    \centering
    \begin{tabular}{|c|c|c|c|}
    \hline
    \rule{0pt}{3ex}
        \cellcolor[gray]{.8}problem&\cellcolor[gray]{.8}\begin{tabular}{c}protocol\\type\end{tabular}&\cellcolor[gray]{.8}beeping rounds&\cellcolor[gray]{.8}ref\\
    [.03in]
    \hline
    \rule{0pt}{3ex}
        \multirow{5}{*}{\begin{tabular}{c}simulation of one\\\congest round\end{tabular}}&\multirow{2}{*}{randomized}&$O(\Delta\min(n,\Delta^2)\log n)$ whp&\cite{ashkenazi2020brief}\\
    [.03in]
    \cline{3-4}
    \rule{0pt}{3ex}
        &&$O(\Delta^2\log n)$ whp&\cite{davies2023optimal}\\
    [.03in]
    \cline{2-4}
    \rule{0pt}{3ex}
        &\multirow{2}{*}{deterministic}&$O(\Delta^4 \log n)$&\cite{beauquier2018fast}\\
    [.03in]
    \cline{3-4}
    \rule{0pt}{3ex}
        &&\cellcolor[gray]{.9}$O(\Delta^2 \polylog n)$&{\bf Thm.~\ref{thm:congest-sim}}\\
    [.03in]
    \cline{2-4}
    \rule{0pt}{3ex}
        &rand./det.&\cellcolor[gray]{.9}$\Omega(\Delta^2\log n)$&\cite{davies2023optimal}\\
    [.03in]
%    \cline{1-1}\cline{3-4}
    \hline
    \hline
    \rule{0pt}{3ex}
        \multirow{2}{*}{\begin{tabular}{c}$B$-bit $h$-hop simulation\end{tabular}}&deterministic&\cellcolor[gray]{.9}$O(h\cdot B\Delta^{h+2}\polylog n)$ &{\bf Thm.~\ref{thm:multihopub}}\\
    [.03in]
    \cline{2-4}
    \rule{0pt}{3ex}
        &rand./det.&\cellcolor[gray]{.9}$\Omega(B\Delta^{h+1})$ &{\bf Thm.~\ref{thm:multihoplb}}\\
    [.03in]
    \hline
    \hline
    \rule{0pt}{3ex}
        \multirow{2}{*}{Local Broadcast}&rand./det.&
        %\cellcolor[gray]{.9}
        $\Omega(\Delta\log n)$&\cite{davies2023optimal}\\
    [.03in]
    \cline{2-4}
    \rule{0pt}{3ex}
%        &\multirow{2}{*}{deterministic}&
        %\cellcolor[gray]{.9}
        &deterministic
%        $\Omega(\min\{n,\Delta^2/\log^2 n\})$&[Cor.~\ref{cor:LBlower}]\\
%    [.03in]
%    \cline{3-4}
%    \rule{0pt}{3ex}
        &$O(B\Delta^2\log n)$, for $B$ bits&{\bf Thm.~\ref{th:local_broadcast}}\\
    [.03in]
    \hline
    \hline
    \rule{0pt}{3ex}
        \begin{tabular}{c}$B$-bit $h$-hop\\ Local Broadcast\end{tabular}&rand./det.&$\Omega(B\Delta^{h})$ &{\bf Thm.~\ref{thm:multihoplb}}\\
    [.03in]
    \hline
    \hline
    \rule{0pt}{3ex}
        Learning Neighborhood&\multirow{5}{*}{deterministic}&$O(\Delta^2 \log^2 n)$&{\bf Thm.~\ref{th:learning_neighbourhood}}\\
    [.03in]
    \cline{1-1}\cline{3-4}
    \rule{0pt}{3ex}
        Cluster Gathering&&$O(\Delta^2 \log^4 n)$&{\bf Thm.~\ref{th:cluster_gathering}}\\
    [.03in]
    \cline{1-1}\cline{3-4}
    \rule{0pt}{3ex}
        \begin{tabular}{c}$(\log n,\log^2 n)$-Network\\Decomposition\end{tabular}&&$O(\Delta^2 \log^8 n)$&{\bf Thm.~\ref{thm:local-decomposition}}\\
    [.03in]
    \cline{1-1}\cline{3-4}
    \rule{0pt}{3ex}
        \multirow{2}{*}{MIS}&&\cellcolor[gray]{.9}$O(\Delta^3+\Delta^2\log n)$&\cite{beauquier2018fast}\\
    [.03in]
    \cline{3-4}
    \rule{0pt}{3ex}
        &&\cellcolor[gray]{.9}$O(\Delta^2 \polylog n)$&{\bf Cor.~\ref{cor:mis}}\\
%    [.03in]
%    \cline{2-4}
%    \rule{0pt}{3ex}
%        &randomized&$O(\log^3 n)$ whp&\cite{afek2013beeping}\\
    [.03in]
    \hline
    \end{tabular}
    \caption{
    Summary of our results and related work for Beeping Networks with $n$ nodes and maximum degree $\Delta$. All protocols are distributed. The local broadcast lower bound applies 
    %regardless of adaptiveness.
    \mam{to non-adaptive content-oblivious protocols.}
    %Highlighted pairs of cells (in the parts describing general simulations and local broadcast) show our main deterministic results and their comparison with the existing randomized results.
    Highlighted pairs of cells show a comparison of our main results. For the general \congest simulator and Local Broadcast we show their comparison with the existing randomized results, whereas for MIS we show the upper bound improvement with respect to existing~work.
    }
    \label{table:relwork}
    \vspace{-3ex}
\end{table}

%% file: related.tex
%\vspace*{-1ex}
\subsection{Related Work}
\label{sec:relwork}

%\pga{[NOTES: Compare our results to deterministic results and another comparison with randomized results (Davies' paper). DONE
%Mention that Theorem 8 (lower bound) only applies to content-oblivious protocols, but in the beeping model all protocols are content-oblivious. DONE ]}

%A variety of computational problems that are fundamental for distributed computing in communication networks has been studied for Beeping Networks. A summary of the closest related results in comparison with our main results is shown in Table~\ref{table:relwork}.

The BN model was defined by Cornejo and Kuhn in~\cite{cornejo2010deploying} in 2010, inspired by 
continuous beeping studied by Degesys et al.~\cite{degesys2007desync} and Motskin et al.~\cite{motskin2009lightweight}, and by the implementation of coordination by carrier sensing given by Flury and Wattenhofer in~\cite{flury2010slotted}.
Since then, the literature has included studies on 
MIS and Coloring~\cite{afek2011biological,afek2013beeping,jeavons2016feedback,holzer2016brief,beauquier2018fast,casteigts2019design}, 
Naming~\cite{chlebus2017naming}, 
Leader Election~\cite{ghaffari2013near,forster2014deterministic,dufoulon2018beeping}, 
Broadcast~\cite{ghaffari2013near,hounkanli2015deterministic,hounkanli2016asynchronous,czumaj2019communicating,beauquier2019optimal},
and Shortest Paths~\cite{dufoulon2022beeping}.

Techniques to implement \congest algorithms in Beeping Networks were studied. In~\cite{beauquier2018fast},
the approach is to schedule transmissions according to a 2-hop $c$-coloring to avoid collisions. The multiplicative overhead introduced by the simulation is in $O(c^2 \log n)$. A constant $c$ is enough for the simulation, but the only coloring algorithm provided in the same paper takes time $O(a^2\Delta^2\log^2 n +a^3 \Delta^3 \log n)$ for a $(\Delta^2+1)$-coloring. Thus, the multiplicative overhead is in $O(\Delta^4 \log n)$. 
%In comparison with our $O(\Delta^2 \polylog n)$ overhead, we improve by a factor of $\Delta^2/\polylog n$.
Thus, with respect to such work, our simulation improves by a factor of $\Delta^2/\polylog n$.

On the side of randomized protocols, in a recent work by Davies~\cite{davies2023optimal}, a protocol that simulates a \congest round in a Beeping Network with $O(\Delta^2\log n)$ overhead is presented. The protocol works even in the presence of random noise in the communication channel. Still, it is correct only with high probability (whp),\footnote{An event $E$ occurs \emph{with high probability} if $Prob(E)\geq 1-1/n^c$ for some $c>0$.} and requires a polynomial number of \congest rounds in the simulated algorithm.   
More relevant for comparison with our work, in the same paper, a lower bound of $\Omega(\Delta^2 \log n)$ on the overhead to simulate a \congest round is shown. The lower bound applies even in a noiseless environment and regardless of randomization. 
Thus, our simulation is optimal modulo some poly-logarithmic factor. 
%\mam{The proof of the lower bound in~\cite{davies2023optimal} is based on requiring each node to transmit a different string of bits to each neighbor. If the string of bits could be the same (a much simpler problem known as local broadcast), that lower bound would collapse to $\Omega(\Delta \log n)$. In this work, we prove a $\Omega\left(\min\left\{n,\Delta^2/\log^2 n\right\}\right)$ lower bound for local broadcast.}
%
Another randomized simulation was previously presented in~\cite{ashkenazi2020brief} with overhead of $O(\Delta\min(n,\Delta^2)\log n)$ whp.

%Another simulation of \congest algorithms in Beeping Networks is presented in~\cite{ashkenazi2020brief}. The multiplicative overhead is $O(\Delta\min(n,\Delta^2)\log n)$, but the simulation succeeds only whp.

With respect to our study on MIS, the closest work is the deterministic protocol presented in~\cite{beauquier2018fast}, which runs in $O(\Delta^2 \log n + \Delta^3)$. Thus, our results improve by a factor of $\Delta/\polylog n$ for any $\Delta\in \omega(\log n)$, and match the running time (modulo poly-logarithmic factor) for $\Delta\in O(\log n)$. 

On the side of randomized MIS protocols, an upper bound of $O(\log^3 n)$ has been shown in~\cite{afek2013beeping} for the same beeping model, and faster with some additional assumptions.

Our network decomposition algorithm for Beeping Networks is heavily based on the protocol for the \congest model~\cite{ghaffari2021improved} that shows a $(\log n, \log^2 n)$ network decomposition in $O(\log^5 n)$ \congest rounds. For comparison, we complete the same network decomposition in $O(\Delta^2 \log^7 n)$ beeping rounds.

Most works in the Beeping Networks literature assume that all nodes start execution simultaneously, called global synchronization. For settings where that is not the case, it is possible to simulate global synchronization in Beeping Networks where nodes start at different times, as shown in~\cite{afek2013beeping,forster2014deterministic,dufoulon2018beeping,hounkanli2020global}.

%% file: model.tex
\newcommand{\myboldmath}{}%\boldmath doesn't work with latex
\newcommand{\defn}[1]{{\textit{\textbf{\myboldmath #1}}}}

\vspace*{-1.5ex}
\section{Model, Notation, and Problems}
\label{sec:model}

\vspace*{-1ex}
We consider a communication network formed by $n$ devices with communication and computation capabilities, called \defn{nodes}. 
Each node has a unique \defn{ID}
% \pga{. Depending on the problem, we assume that the IDs come from the range $[1,n]$ or $[1,n^c]$} 
from the range $[1,n^c]$
for some constant $c \geq 1$.\footnote{The availability of identifiers is essential in order to break symmetry in deterministic protocols, as pointed out in previous works on deterministic protocols in the Beeping model~\cite{beauquier2018fast,dufoulon2018beeping}.}
Nodes communicate by sending \defn{messages} among them. 
A message is composed of a binary sequence containing the source node ID, the destination node ID (if applicable), and the specific information to be sent. 
If the destination node receives the message from the source node, we say that the message was \defn{delivered}.
Each pair of nodes that are able to communicate directly (i.e., without relaying communication through other nodes) are said to be connected by a communication \defn{link} and are called \defn{neighbors}.
We assume that links are \defn{symmetric}, i.e., messages can be sent in both directions (delivery is restricted to the communication models specified below).
The network topology defined by the communication links is modeled with an undirected graph $G=(V,E)$ where $V$ is the set of nodes and $E$ is the set of links. 
\mam{If $E$ is such that for every pair of nodes $u,v\in V$ there is a path of links connecting $u$ and $v$ we say that the network is \defn{connected}.}
For each node $v\in V$, the set of neighbors of $v$ is called its \defn{neighborhood}, denoted as $N(v)$.
We assume that time is slotted in \defn{rounds} of communication. All nodes start running protocols simultaneously, i.e., the network is \defn{synchronous}.
We assume that computations take negligible time with respect to communication.
Thus, we measure algorithm performance in rounds.

\vspace*{-2ex}
\subsection{Communication Models}

\vspace*{-0.95ex}
\parhead{Beeping Networks~\cite{cornejo2010deploying}}
%\paragraph{Beeping Networks~\cite{cornejo2010deploying}:}
In this model, 
in each round each node can either \defn{beep} (send a signal) or \defn{listen} (do not send any signal). 
By doing so, nodes obtain the following \defn{communication channel feedback}. 
In any given round, a listening node \defn{hears} either \defn{silence} (no neighbor beeps) or \defn{noise} (one or more neighbors beep).
A listening node that hears noise cannot distinguish between a single beep and multiple beeps. 

Network protocols may use the channel feedback
(i.e. the temporal sequence of strings from \{``silence'',``noise''\})
to make decisions adaptively.
However, delivering messages is not straightforward because it requires sending (and receiving) the whole binary sequence of the message (according to some beeping schedule, possibly changing adaptively during the communication), somehow encoded with beeps. In that sense, protocols for Beeping Networks can be seen as radio network coding to cope with the communication restrictions (as in~\cite{efremenko2018interactive} to cope with noise). 

\vspace*{-2ex}
\begin{figure}[th]
\centering
\begin{subfigure}[htbp]{0.3\textwidth}
\centering
\vspace*{-2ex}
\includegraphics[scale=0.12]{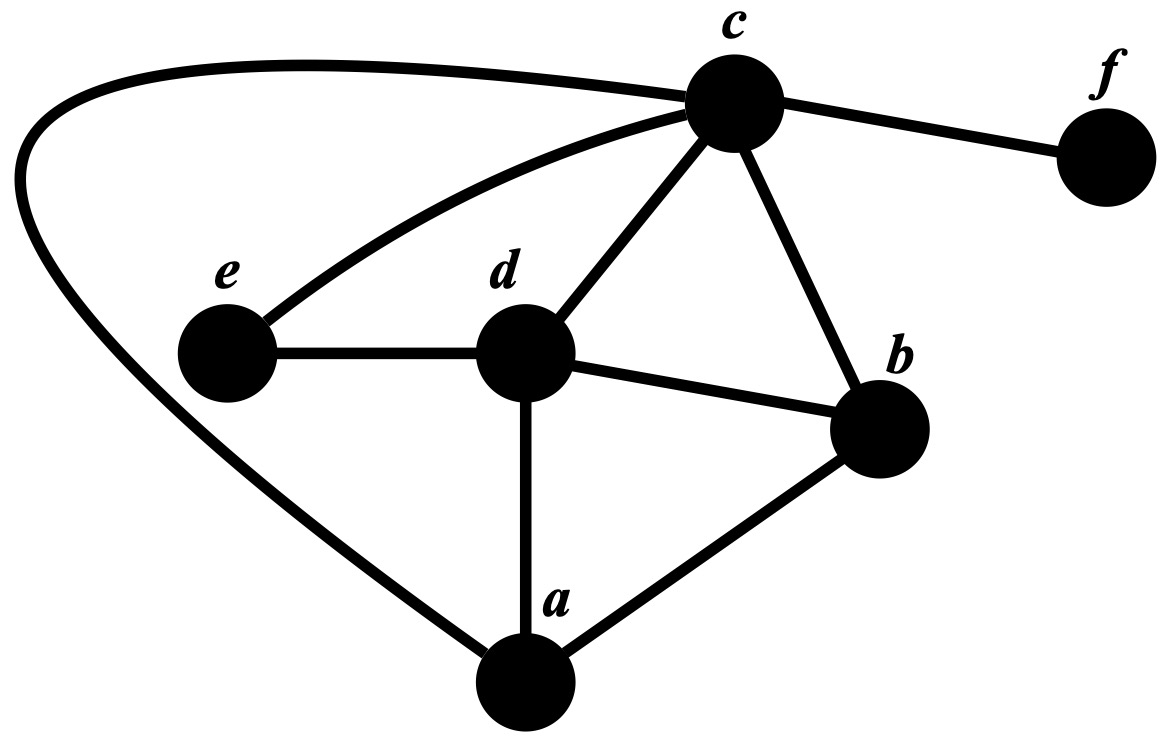}
\caption{All nodes listen. All nodes hear silence.}
\label{subfig:network}
\end{subfigure}
\hspace{0.1in}
\begin{subfigure}[htbp]{0.3\textwidth}
\centering
\includegraphics[scale=0.12]{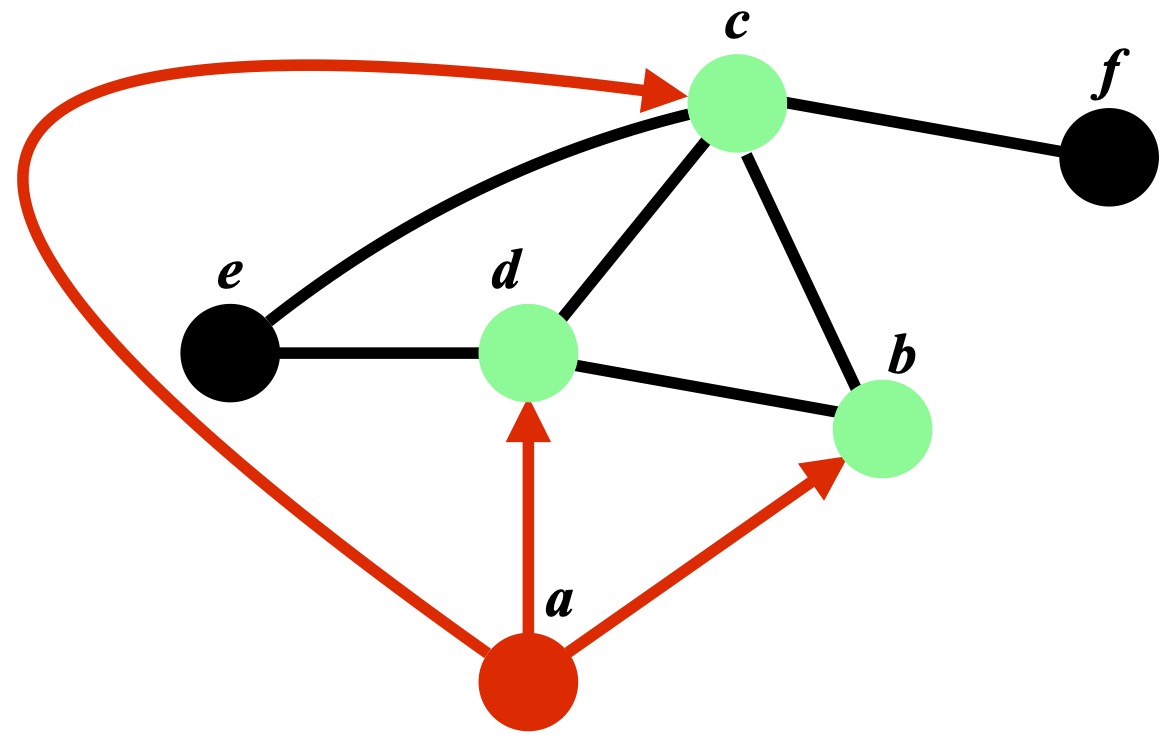}
\caption{Node $a$ beeps, $\{b,c,d,e,f\}$ listen. $\{b,c,d\}$ hear noise. $\{e,f\}$ hear silence.}
\label{subfig:1beep}
\end{subfigure}
\hspace{0.1in}
\begin{subfigure}[htbp]{0.3\textwidth}
\centering
\vspace*{-0.2ex}
\includegraphics[scale=0.12]{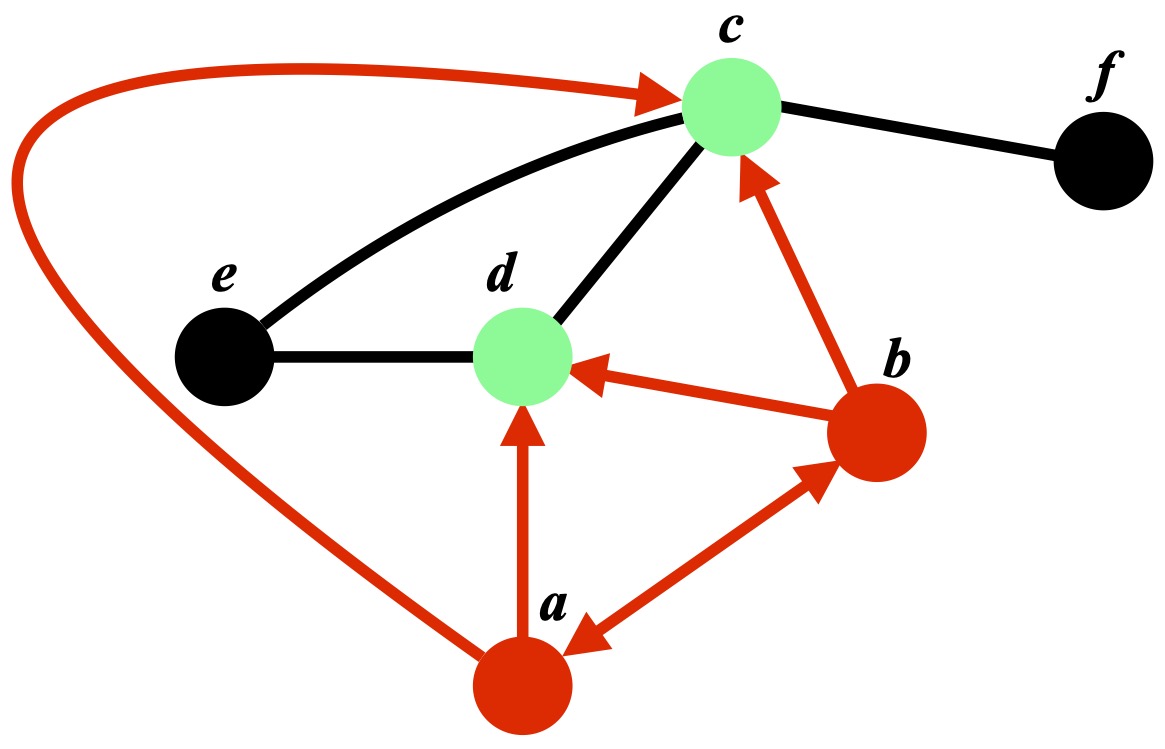}
\caption{Nodes $\{a,b\}$ beep, $\{c,d,e,f\}$ listen. $\{c,d\}$ hear noise. $\{e,f\}$ hear silence.}
\label{subfig:2beep}
\end{subfigure}
\vspace*{-1.5ex}
\caption{Beeping Network communication model example.}
\vspace*{-1.5ex}
\label{fig:BNmodel}
\end{figure}

\remove{%%%%%%%  START  REMOVE  %%%%%%%%
\parhead{Radio Networks with Collision Detection~\cite{chlamtac1985broadcasting}}
%\paragraph{Radio Networks with Collision Detection~\cite{chlamtac1985broadcasting}:}
In this model, 
in each round each node can either \defn{transmit} (send a message) or \defn{receive} (do not send any message). 
By doing so, nodes obtain the following \defn{communication channel feedback}. 
In any given round a receiving node \defn{receives} either \defn{silence} (no neighbor transmits), or a \defn{successful transmission} (one neighbor transmits), or a \defn{collision} (two or more neighbors transmit).

Note, that having three states the channel feedback of this model is richer, potentially allowing protocols that adapt more effectively. Moreover, the length of the messages transmitted is not restricted, reducing the problem of message delivery to avoiding collisions. 

}%%%%%%%  END  REMOVE  %%%%%%%%%%%%%%%%%

\parhead{\congest Networks~\cite{peleg2000distributed}}
%\paragraph{\congest Networks~\cite{peleg2000distributed}:}
In this model, 
in each round each node can send a (possibly different) message of $O(\log n)$ bits to each neighbor independently. All nodes receive the messages sent by their neighbors. That is, there are no collisions. In the \defn{\congest Broadcast} version of this model, each node can only broadcast the same message to all its neighbors in each round.

\vspace*{-1ex}
\subsection{Problems Studied}

Our main research question in this work is how to efficiently simulate a round of communication of a \congest Network protocol in a Beeping Network. 
We also study several distributed computing problems (of independent interest) in the context of Beeping Networks, and applications of our simulator that improve efficiency with respect to known solutions for Beeping Networks.

Before specifying the problems studied, we define the following notation for any Beeping Network with topology graph $G=(V,E)$.
A \defn{cluster} of nodes $C$ is a subset of nodes in $V$ (i.e. $C\in V$)~\footnote{Note that a cluster does not need to be connected.}.
We say that clusters $C_1$ and $C_2$ are \defn{neighboring clusters} if and only if there exist vertices $v_1 \in C_1$ and $v_2 \in C_2$ such that $\{v_1,v_2\} \in E$.
We say that a subgraph $G'$ of graph $G$ has a \defn{weak-diameter} $D$ if each vertex in $G'$ is at most $D$ distance away from every other vertex in $G$ (the original graph).\footnote{In contrast with the regular diameter, where each vertex in $G'$ is at most $D$ distance away from every other vertex in $G'$.
%(in the \emph{subgraph}).
}
An \defn{independent set} is a set of nodes $S\subseteq V$ such that $\forall u,v \in S : \{u,v\}\notin E$.
A \defn{maximal independent set (MIS)} is an independent set that is not a subset of any other independent set. 

The definitions of all problems studied follow. 
%\textcolor{green}{[[should we require stopping???]]} 

\parhead{Local Broadcast} 
%\paragraph{Local Broadcast:} 
Given a Beeping Network with topology graph $G=(V,E)$, where each node $v\in V$ holds a message $m_v$, this problem is solved once $m_v$ is delivered to every node in $N(v)$. 

\parhead{Learning Neighborhood} 
%\paragraph{Learning Neighborhood:} 
Given a Beeping Network with topology graph $G=(V,E)$, the learning neighborhood problem is solved once every node $v\in V$ knows \mam{the ID of every node $u\in N(v)$.} 

\parhead{Cluster Gathering}
%\paragraph{Cluster Gathering:}
Given a Beeping Network with topology graph $G=(V,E)$, where 
each node $v\in V$ holds some data $d_v$, 
the set of nodes is partitioned in $k\geq 1$ \defn{clusters} as $\{C_1,C_2,\dots,C_k\}$, and 
for each cluster $C_i$, $i\in[k]$, there is a designated \defn{leader} node $l_{C_i}\in C_i$ and a Steiner tree of depth at most $O(\log^2 n)$ that spans the cluster $C_i$,
the cluster gathering problem is solved once all leaders have received the data of all nodes in their cluster. That is, for each $i\in [k]$, $l_i$ knows $d_v$, for all $v\in C_i$.

\parhead{Network Decomposition} 
%\paragraph{Network Decomposition:} 
Given a Beeping Network with topology graph $G=(V,E)$ and parameter integers $C$ and $D$,
the $(C,D)$-network decomposition problem is to find a partition of $V$ into clusters $\{C_1,C_2,\dots\}$ such that each cluster has weak-diameter at most $O(D)$, and each cluster can be assigned a color so that, 
for every pair of neighboring clusters $C_i,C_j$, 
the color of $C_i$ and $C_j$ are different,
and the number of colors used is in $O(C)$.

\parhead{\congest Simulation}
%\paragraph{\congest Simulation:}
Given a Beeping Network with topology graph $G=(V,E)$, where each node $u\in V$ may hold a message $m_{u,v}$ of $O(\log n)$ bits that must be delivered to node $v\in N(u)$, the \congest simulation problem is solved once for every $u\in V$ and $v\in N(u)$, 
$m_{u,v}$ and the ID of $u$ has been delivered~to~$v$.

\parhead{Maximal Independent Set}
%\paragraph{Maximal Independent Set:}
Given a Beeping Network with topology graph $G=(V,E)$, 
The Maximal Independent Set problem is solved when, for some such set,
 $S\subseteq V$, every node $v\in S$, $v$  knows that it is~in~$S$.

%% file: primitives_intro.tex
\vspace*{-1ex}
\section{Initial Results}
\label{sec:primitives}

In this section, we present beeping protocols for four fundamental network problems, usually used as building blocks of more complex tasks. Namely, Local Broadcast, Cluster Gathering, Learning Neighborhood, and Network Decomposition. 
%These protocols will be used as building blocks in our simulation of \congest rounds in beeping networks (Section~\ref{sec:main-simulation}).}
%algorithms that will be used in the network decomposition algorithm, namely, algorithms solving efficiently the problems of learning neighborhood, local broadcast and cluster gathering.
%
The following theorems establish the performance of our protocols. The details of the algorithms as well as the proofs of the theorems are left to Section~\ref{sec:prim_details}.

Recall that IDs of nodes come from the range $[1,n^c]$, for some constant $c \geq 1$.

\begin{restatable}[]{theorem}{localbroadcastthm} %\begin{theorem}
\label{th:local_broadcast}
    Let $\cN$ be a Beeping Network with 
    %set $V$ of 
    $n$ nodes, where each node 
    $v$
    %$v\in V$ 
    knows $n$, parameter~$c$, the maximum degree $\Delta$, and its neighborhood $N(v)$, and holds a message $m_v$ of length at most $B>0$.
    %Then, there is a deterministic distributed local broadcast algorithm that works in $O(k\Delta^2 \log^2 n)$ beeping rounds.
    There is a deterministic distributed algorithm that solves local broadcast on $\cN$ in $O(B\Delta^2 \log n)$ beeping~rounds.
    %\dk{???? SHOULDN'T IT BE $O(\Delta^2 (B+\log n)\log n)$ ????}
%\end{theorem}
\end{restatable}

\vspace*{-1.5ex}
\begin{restatable}[]{theorem}{learningneighthm} %\begin{theorem}
\label{th:learning_neighbourhood}
    Let $\cN$ be a Beeping Network with 
    %set $V$ of 
    $n$ nodes, where each node 
    %$v\in V$ 
    $v$
    knows $n$ and parameter~$c$.
    %There is a deterministic distributed learning-neighborhood algorithm that works in $O(\Delta^2 \log^2 n)$ beeping rounds.
    There is a deterministic distributed algorithm that solves learning neighborhood  on $\cN$ in $O(\Delta^2 \log^2 n)$~beeping~rounds.
%\end{theorem}
\end{restatable}

\vspace*{-1.5ex}
\begin{restatable}[]{theorem}{clustergatherthm} %\begin{theorem}
\label{th:cluster_gathering}
    Let $\cN$ be a Beeping Network with %set $V$ of 
    $n$ nodes, where each node 
    $v$
    %$v\in V$ 
    knows $n$, parameter $c$, the maximum degree $\Delta$, and its neighborhood $N(v)$.
    %There is a deterministic distributed cluster gathering algorithm that works in $O(\Delta^2 \log^4 n)$ beeping rounds.
    There is a deterministic distributed algorithm that solves cluster gathering on $\cN$ in $O(\Delta^2 \log^4 n)$ beeping rounds.
%\end{theorem}
\end{restatable}

% \pga{In the next theorem, we assume that node IDs come from range $[1,n]$.}

\vspace*{-1.5ex}
\begin{restatable}[]{theorem}{networkdecompthm} %\begin{theorem}
\label{thm:local-decomposition}
    Let $\cN$ be a Beeping Network with set $V$ of $n$ nodes, where each node $v\in V$ knows $n$, parameter $c$ and the maximum degree $\Delta$.
    There is a deterministic distributed algorithm that computes a $(\log n, \log^2 n)$-network decomposition of $\cN$ in $O(\Delta^2 \log^8 n)$ beeping rounds.
%\end{theorem}
\end{restatable}

% \pga{[TODO: Check if GGR algorithm can really be adapted to $n^c$ IDs trivially.]}

%% file: algorithm.tex
\begin{algorithm}[b!]
\DontPrintSemicolon
\SetKwFunction{ctob}{C2B}
\SetKwFunction{announcer}{announcer}
\SetKwFunction{listener}{listener}
\SetKwProg{myalg}{Algorithm}{}{}
\let\oldnl\nl
\newcommand{\nlnonumber}{\renewcommand{\nl}{\let\nl\oldnl}}
\nlnonumber
$E(u)\leftarrow \{\{u,v\}|\{u,v\}\in E$ for some $v\in V\}$\;
\nlnonumber
\myalg{\ctob{}}{
\For{each epoch $i=1,2,\dots,\log\Delta$}{
    $k_i\leftarrow\Delta/2^i$\;
    \For{each phase $j=1,2,\dots,|\mF_{\Delta,k_i}|$}{ 
        \eIf{$u\in \mF_{\Delta,k_i}(j)$}
        {\announcer{$u,k_i$}}
        {\listener{$u,k_i$}}
    }
}
}
\caption{\alg algorithm for each node $u$. $E(u)$ is a global variable.} 
\label{algC2Bv2}
\end{algorithm}

%% file: algannouncer.tex
\begin{algorithm}[t!]
\DontPrintSemicolon
\SetKwFunction{announcer}{announcer}
\SetKwProg{myalg}{Procedure}{}{}
\let\oldnl\nl
\newcommand{\nlnonumber}{\renewcommand{\nl}{\let\nl\oldnl}}
\nlnonumber
\myalg{\announcer{$v,k_i$}}{
    \tcp{announcing super-round}
    \For{each round $r=1,2,\dots,2\log n$}{
        \leIf{$\langle v\rangle(r)=1$}{beep}{listen}
    }
        \For{each sub-phase $a=1,2,\dots,\log k_i$}{\label{line:subphaseloopA}
            \For{$b=1,2,\dots,|\mF_{k_i/2^{a-2},k_i/2^{a-1}}|$}{
                \tcp{responding 3 super-rounds}
                \lFor(\tcp*[h]{announcer only listens}){$6\log n$ rounds}{listen}
                \eIf{some $\langle w\rangle\langle v\rangle\langle m_{w,v}\rangle$ was heard {\bf and} 
                $\{w,v\}\in E(v)$}{
                    \tcp{confirming 3 super-rounds}
                    \For{each round $r=1,2,\dots,2\log n$}{
                        \leIf{
                        $\langle v\rangle(r)=1$}
                        {beep}{listen}
                    }
                    \For{each round $r=1,2,\dots,2\log n$}{
                        \leIf{
                        $\langle w\rangle(r)=1$}
                        {beep}{listen}
                    }
                    \For{each round $r=1,2,\dots,2\log n$}{
                        \leIf{
                        $\langle m_{v,w}\rangle(r)=1$}
                        {beep}{listen}
                    }
                    $E(v)\leftarrow E(v)\setminus \{w,v\}$ \tcp{link realized}
                    \lIf{$E(v)=\emptyset$}{$v$ stops executing}
                }{
                    \lFor(\tcp*[h]{wait to synchronize}){$6\log n$ rounds}{listen}
                }
            }
        }
    }
\caption{\alg algorithm for \underline{announcer} node $v$.} 
\label{algC2Bv2A}
\end{algorithm}

%% file: alglistener.tex
\begin{algorithm}[t!]
\DontPrintSemicolon
\SetKwFunction{listener}{listener}
\SetKwProg{myalg}{Procedure}{}{}
\let\oldnl\nl
\newcommand{\nlnonumber}{\renewcommand{\nl}{\let\nl\oldnl}}
\nlnonumber
\myalg{\listener{$w,k_i$}}{
    $status(u)\leftarrow$ \texttt{nil}\;
    \tcp{announcing super-round}
    \lFor(\tcp*[h]{listener only listens}){$2\log n$ rounds}{listen}
    \If{some $\langle v\rangle$ was heard {\bf and} $\{v,w\}\in E(w)$}{
        $status(w) \leftarrow v$-$responsive$\;
    }
    \For{each sub-phase $a=1,2,\dots,\log k_i$}{\label{line:subphaseloopL}
        \For{$b=1,2,\dots,|\mF_{k_i/2^{a-2},k_i/2^{a-1}}|$}{
            \eIf{$status(w)=v$-$responsive$ {\bf and} 
                $w\in\mF_{k_i/2^{a-2},k_i/2^{a-1}}(b)$}{
                \tcp{responding 3 super-rounds}
                \For{each round $r=1,2,\dots,2\log n$}{
                    \leIf{$\langle w\rangle(r)=1$}{beep}{listen}
                }
                \For{each round $r=1,2,\dots,2\log n$}{
                    \leIf{$\langle v\rangle(r)=1$}{beep}{listen}
                }
                \For{each round $r=1,2,\dots,2\log n$}{
                    \leIf{$\langle m_{w,v}\rangle(r)=1$}{beep}{listen}
                }
                \tcp{confirming 3 super-rounds}
                \lFor(\tcp*[h]{listener only listens}){$6\log n$ rounds}{listen}
                \If{some $\langle v\rangle\langle w\rangle\langle m_{v,w}\rangle$ was heard}{
                    $status(w)\leftarrow$ \texttt{nil}\;
                    $E(w)\leftarrow E(w)\setminus \{v,w\}$ \tcp{link realized}
                    \lIf{$E(w)=\emptyset$}{$w$ stops executing}
                }            
            }{
                \lFor(\tcp*[h]{wait to synchronize}){$12\log n$ rounds}{listen}
            }
        }
    }
}
\caption{\alg algorithm for \underline{listener} node $w$.} 
\label{algC2Bv2L}
\end{algorithm}

%% file: multihop.tex
%\section{Multi-hop Bounds}
%\vspace*{-3.1ex}
\section{Multi-hop \mam{Simulation}}

% - lower bound of $\Omega(\Delta)$ for Learning neighborhood based on Davies lower bound (Lemma 14 in Davies' paper); write it down!

% - algorithm for multihop Learning neighborhood via flooding (repeated use of Local Broadcast); we can also compute the shortest paths to each node within k hops.

% \dk{!!!! POTENTIAL NOTATION CLASH - $k$ denoted size of messages in previous section(s) while here denotes $k$-hop distance !!!!}

\dk{We generalize the 
%local communication 
\mam{simulation}
at distance $1$ to the following {\em $B$-bit $h$-hop simulation} problem: each node has messages, potentially different, of size at most $B$ \mam{bits} addressed to any other node, and it needs to deliver them to all destination nodes \mam{within} distance at most $h$ \mam{hops}. 
If each node has only a single message of size at most $B$ \mam{bits} to be delivered to all nodes \mam{within} distance at most $h$ \mam{hops}, then we call this restricted version {\em $B$-bit $h$-hop Local Broadcast}.
Note also that we do not require messages addressed to nodes of distance larger than $h$ to be delivered.}
Below we generalize the lower bound for single-hop simulation %by Davies
in~\cite{davies2023optimal} 
to multi-hop simulation and multi-hop local broadcast.

\begin{theorem}\label{thm:multihoplb}
    There 
    %exists 
    is
    an adversarial network of size $\Theta(\Delta^h)$ such that any $B$-bit $h$-hop simulation algorithm requires $\Omega(B\Delta^{h+1})$ beeping rounds to succeed with probability more than $2^{-\frac{1}{2}\cdot B(\Delta-1)^{h-2}(\Delta/2)^3} = 2^{-\Theta\left(B\Delta^{h+1}\right)}$.

    There exists an adversarial network of size $\Theta(\Delta^h)$ such that any $B$-bit $h$-hop Local Broadcast requires $\Omega(B\Delta^{h})$ beeping rounds to succeed with probability more than $2^{-\frac{1}{2}\cdot B(\Delta-1)^{h-2}(\Delta/2)^2} = 2^{-\Theta\left(B\Delta^{h}\right)}$.
\end{theorem}
\begin{proof}
    \noindent\textbf{Problem instance.}
    \mamr{We describe the construction of the adversarial network and input set of messages used to prove our lower bound as follows (refer to Figures~\ref{fig:multihop_graph} and~\ref{fig:multihop_graph2}).}
    Consider a full bipartite graph $K_{\Delta/2,\Delta/2}$, with one part called $T$ and the other $R$. We 
    %will 
    focus on transmissions going towards nodes in $R$, hence nodes in $R$ will be called receivers, while nodes in $T = T_1$ will be called the first \mamr{of $h$ layers}~of~transmitters. 
    
    % Let us ignore for a moment $R$ and its links. 
    Each node in $T_1$ will be a root of an $(h-1)$-depth tree of transmitters. We create a second layer of transmitters $T_2$ composed of $(\Delta/2)^2$ nodes. Each node in $T_1$ \mamr{(already connected to each node in $R$)} will also be connected to different $\Delta/2$ nodes in $T_2$.
    % , thus each node in $T$ has degree $\Delta/2 + \Delta/2 = \Delta$. 
    %Every
    \mamr{For subsequent layers of transmitters, that is each} layer $T_i$, for $3 \leq i \leq h$, will be composed of $(\Delta/2)^2 (\Delta-1)^{i-2}$ nodes. Each node in layer $T_{i-1}$, for $3 \leq i \leq h$, will be connected to different $\Delta-1$ nodes in layer~$T_i$. 
    %\pga{See Figure~\ref{fig:multihop_graph}.}
    % Thus, nodes in layer $T_{i-1}$ for $3 \leq i \leq k-2$ have degree $1 + \Delta-1 = \Delta$. Finally, nodes in layer $T_{k-1}$ have degree $1$.
%
    \begin{figure}[t]
        \centering
        \includegraphics[width=0.5\linewidth]{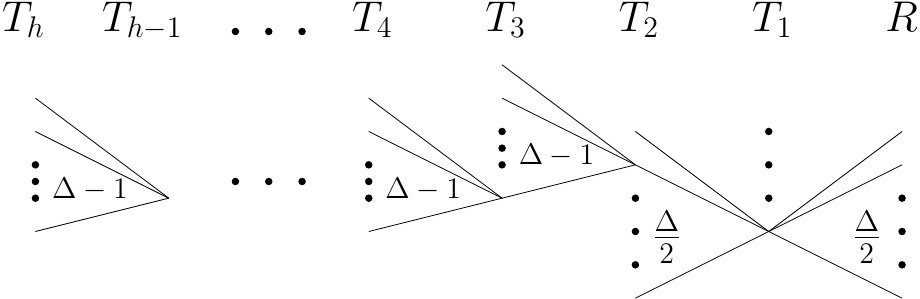}
        \caption{An illustration of the structure of the graph. The graph is partitioned into vertical layers $T_h$, \dots, $T_1$, $R$. \pga{The graph is branching out heavily, so we show only a path} from an arbitrary node in $T_1$ layer to an arbitrary node in $T_h$ layer, with all the edges incident to the path. The numbers between the layers denote the number of edges between the layers that are incident \pga{to the path or on} the path. Recall that layers $T_1$ and $R$ have $\Delta/2$ nodes each, while the other layers have significantly more nodes, but we only show nodes that are adjacent to the considered path.}
        \label{fig:multihop_graph}
    \end{figure}
    Note that each node \mamr{in the defined network} has at most~$\Delta$~neighbors.

    %We have defined the graph we will be working on. Now we will describe the transmissions. 
    \mamr{We define now the input set of messages as follows.}
    Let each node $v \in T_{h}$ have a $B$-bit message $m_{v \rightarrow u}$ to each node $u \in R$. We choose those messages uniformly at random. We will show that just these messages cannot be relayed efficiently and we do not need any other messages in our problem instance\footnote{Alternatively we can make all the other messages known to the optimal algorithm, e.g., by setting them to be $0^B$.}.

%    \noindent\textbf{Gossiping.} Here we analyze the Gossiping algorithms.
    \noindent\textbf{\mam{Multihop Simulation.}} Here we analyze the \mam{multihop simulation} algorithms.

    There are $(\Delta/2)^2(\Delta-1)^{h-2}$ nodes in $T_h$, and each of them has $\Delta/2$ (possibly different) messages, one for each node in $R$. Therefore, there are $(\Delta/2)^3 (\Delta-1)^{h-2} \pga{=} \Theta(\Delta^{h+1})$ messages to nodes in $R$ that are passing through nodes in $T_1$.

    % Note that each node in $T_1$ will have to relay $(\Delta-1)^{k-1}\Delta/2 = \Theta(\Delta^k)$ messages to each node in $R$, i.e., $(\Delta-1)^{k-1}(\Delta/2)^2 = \Theta(\Delta^{k+1})$ messages total.
    
    % Note that each of the $\Delta/2$ nodes in $T_1$ will have to relay messages from $\frac{(\Delta/2)^2(\Delta-1)^{k-2}}{\Delta/2}$ nodes in $T_{k}$. Each node in $T_{k}$ has $\Delta/2$ (possibly different) messages, one for each node in $R$. Therefore, there are $(\Delta-1)^{k-2}(\Delta/2)^2 = \Theta(\Delta^{k})$ messages to nodes in $R$.

    Let $\mathcal{R}$ be the concatenated string of local randomness in all the nodes in $R$.
    %in $T_i$, for $1 \leq i \leq h-1$. 
    The output of any receiver $u \in R$ must depend only on $\mathcal{R}$, node IDs and the pattern of beeps and silences of nodes in $T_1$.

    There are $2^t$ possible patterns of beeps and silences in $t$ rounds. Therefore, the output of nodes in $R$ must be one of the $2^t$ possible distributions, where a distribution is over the randomness of $\mathcal{R}$. The correct output of nodes in $R$ is a string $\{0,1\}^{B(\Delta-1)^{h-2}(\Delta/2)^3}=\{0,1\}^{\Theta(B\Delta^{h+1})}$ chosen uniformly at random (since the input messages of nodes in $T_{h}$ were chosen uniformly at random). Therefore, the probability of picking the correct result is at most $2^{t-B(\Delta-1)^{h-2}(\Delta/2)^3}$,
    %. Any 
    \mamr{and any} algorithm that finishes within $t \leq \frac{1}{2}\cdot B(\Delta-1)^{h-2}(\Delta/2)^3$ rounds has at most $2^{-\frac{1}{2}\cdot B(\Delta-1)^{h-2}(\Delta/2)^3}$ probability of outputting the correct answer.
    
    \noindent\textbf{Local Broadcast.} The analysis of Local Broadcast is analogous to the analysis of \mam{multihop simulation}, except that there are $\Delta/2$ times fewer messages to transmit \mamr{(because the same message is transmitted to all nodes located within distance $h$ hops of the transmitter)}. The full analysis of Local Broadcast~is~below.

    There are $(\Delta/2)^2(\Delta-1)^{h-2}$ nodes in $T_h$ and each of them has $1$ message to nodes in $R$. Therefore, there are $(\Delta-1)^{h-2}(\Delta/2)^2 \pga{=} \Theta(\Delta^{h})$ messages to nodes in $R$ that are passing through nodes in $T_1$.

    Let $\mathcal{R}$ be the concatenated string of local randomness in all the nodes in $R$. The output of any receiver $u \in R$ must depend only on $\mathcal{R}$, node IDs, and the pattern of beeps and silences of nodes in $T_1$.

    There are $2^t$ possible patterns of beeps and silences in $t$ rounds. Therefore, the output of nodes in $R$ must be one of the $2^t$ possible distributions, where a distribution is over the randomness of $\mathcal{R}$. The correct output of nodes in $R$ is a string $\{0,1\}^{B(\Delta-1)^{h-2}(\Delta/2)^2}=\{0,1\}^{\Theta(B\Delta^{h})}$ chosen uniformly at random (since the input messages of nodes in $T_{h}$ were chosen uniformly at random). Therefore, the probability of picking the correct result is at most $2^{t-B(\Delta-1)^{h-2}(\Delta/2)^2}$,
    %. Any 
    \mamr{and any} algorithm that finishes within $t \leq \frac{1}{2}\cdot B(\Delta-1)^{h-2}(\Delta/2)^2$ rounds has at most $2^{-\frac{1}{2}\cdot B(\Delta-1)^{h-2}(\Delta/2)^2}$ probability of outputting the correct answer.
\end{proof}

    \begin{figure}[t]
        \centering
        \includegraphics[width=1\linewidth]{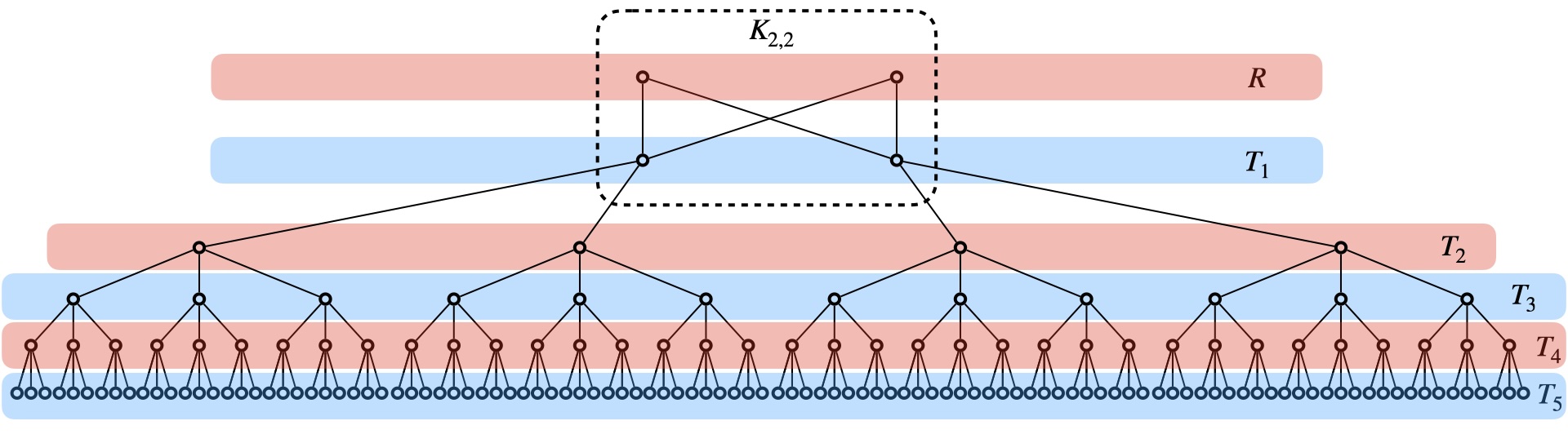}
        \caption{Illustration for Theorem~\ref{thm:multihoplb}. Example of adversarial graph for $\Delta=4$ and $h=5$.}
        \label{fig:multihop_graph2}
    \end{figure}
    
\vspace*{-1.5ex}
\paragraph{Algorithm.}
\pga{A simple algorithm would repeatedly use a 1-hop Local Broadcast routine to flood the network with the messages until nodes at a distance $h$ received the messages. This, however, can take $\Omega(\Delta^{2h})$ rounds. Instead, we limit the flooding by only sending messages along the shortest paths to their destinations using a 1-hop simulation algorithm. The details of the algorithm as well as its analysis are presented next.}

In the beginning, 
%the
nodes use a standard protocol to disseminate their IDs up to distance $h$. They do it in $h$ subsequent epochs, each epoch $i$ of $t_i$ rounds sufficient to run our  Local Broadcast from Section~\ref{sec:primitives} (see Theorem~\ref{th:local_broadcast}) for messages of size $\Delta^i\log n$.
These messages contain different IDs learned by the node at the beginning of the current epoch.
A direct inductive argument, also using the property that there are at most $\Delta^i$ nodes at a distance at most $i$, shows that at the end of epoch $i$, each node knows the IDs of all nodes at a distance at most $i$ from it.
Additionally, each node records in which epoch $i$ it learned each known ID $v$ for the first time and from which of its neighbors $w$ -- and stores this information as a triple $(v,w,i)$.
The invariant for $i=h$ proves that at the end of epoch $h$, each node knows IDs of all nodes of distance at most $h$ from it. The round complexity is %clearly 
% $O(\Delta^{k}\log n \cdot \Delta^2\log^2 n)\le O(\Delta^{k+2}\log^3 n)$
$\sum_{i=1}^h \Delta^{i}\log n \cdot \Delta^2 \log n \pga{=} O(\Delta^{h+2} \log^2 n)$, and as will be seen later, it is subsumed by the round complexity of the second part of our algorithm (as the $\polylog n$ function in Theorem~\ref{thm:congest-sim} is asymptotically bigger than $\log^2 n$).

% NOT SURE IF THIS PARAGRAPH IS NEEDED:
Note that a sequence of triples $(v,w_1=v,1),\ldots,(v,w_{\ell},\ell)$,
stored at nodes $w_2,w_3,\dots,w_\ell,w_{\ell+1}=u$ respectively, represents a shortest path to node $v$ starting from the node $u$; the length of that path~is~$\ell$.

In the second part, nodes also proceed in epochs, but this time each epoch $i$ takes $t^*_i$ rounds sufficient to execute $1$-hop simulation algorithm from Section~\ref{sec:main-simulation} (see Theorem~\ref{thm:congest-sim}) for point-to-point messages of size $(B+\log n)\Delta^h$. 
Here $B$ denotes the known upper bound on the size of any input message.
In epoch $i$, every node $u$ transmits a (possibly different) message of size $(B+\log n)\Delta^h$ to each neighbor $w$. Such a message contains all the input messages of nodes within $i-1$ distance and the recipients of these messages such that $w$ is the next node on the saved shortest path to the recipient. The messages have already traveled $i-1$ hops, so their destination is at most $h-(i-1)$ hops away. More specifically, the message from node $u$
addressed to a neighbor $w$ in epoch $i$ 
contains pairs $(v,m_{z \rightarrow v})$, where $v$ is such that $(v,w,i')$ is stored at the node for some $i'\le h-(i-1)$ and $m_{z \rightarrow v}$ is a message received by the node $u$ in epoch $i-1$ (in case of $i-1=0$, it is the original message of the node addressed to $v$).
A direct inductive argument shows that at the end of each epoch $i$ a node knows at most $\Delta^i$ messages addressed to any node $v$ of distance $\ell \le h-i$ from the node. This invariant is based on the following arguments: 
\begin{itemize}
\item 
Because there is a unique neighbor $w$ of the node $u$ such that a triple $(v,w,\ell)$ is stored at the node, 
% \item 
the number of such nodes $v$ of distance at most $h-i+1$ from the node $u$ is at most $\Delta^{h-i+1}$, 
\item \pga{by the end of epoch $i-1$, node $u$ could receive messages to be relayed to $v$ from $\Delta^{i-1}$ different nodes at distance $i-1$,}
\item \pga{each message contains up to $B$-bit long original message and an ID of length $\log n$,}
\item 
hence, messages of size at most $(B+\log n)\Delta^{i-1} \cdot \Delta^{h-i+1} = (B+\log n) \Delta^h$ are being sent to each neighbor in epoch $i$, and by definition -- epoch $i$ has sufficient number of rounds to deliver them. 
\end{itemize}
The invariant for $i=h$ proves the desired property of $B$-bit $h$-hop simulation. The total number of rounds~is 
\[
O(h\cdot (B+\log n)\Delta^h \cdot \Delta^2\polylog n) 
%\le
\subseteq
O(h\cdot B\Delta^{h+2} \polylog n) 
\ ,
\]
where factor $h$ comes from the number of epochs, each sending at most $\Delta$ point-to-point messages of size at most $(B+\log n)\Delta^h$ to neighbors (by the invariant) using the $1$-hop simulation protocol with overhead $O(\Delta^2 \polylog n)$ (by Theorem~\ref{thm:congest-sim}).
Hence we proved the following.

\begin{theorem}\label{thm:multihopub}
There is a distributed deterministic algorithm solving the $B$-bit $h$-hop simulation problem \mam{in a beeping network} in $O(h\cdot B\Delta^{h+2}\polylog n)$ rounds.
\end{theorem}

%% file: simulator_proofs.tex
%\section{Missing proofs from Section~\ref{sec:main-simulation}}
%\section{Proofs from Section~\ref{sec:main-simulation} -- Analysis of Main Algorithm}
\section{Details of Section~\ref{sec:main-simulation} -- Analysis of the \alg Algorithm}
\label{sec:proofs-main-simulation}

\begin{figure}[thbp]
\centering
\begin{subfigure}[htbp]{0.30\textwidth}
\centering
\vspace*{-10ex}
\includegraphics[width=\linewidth]{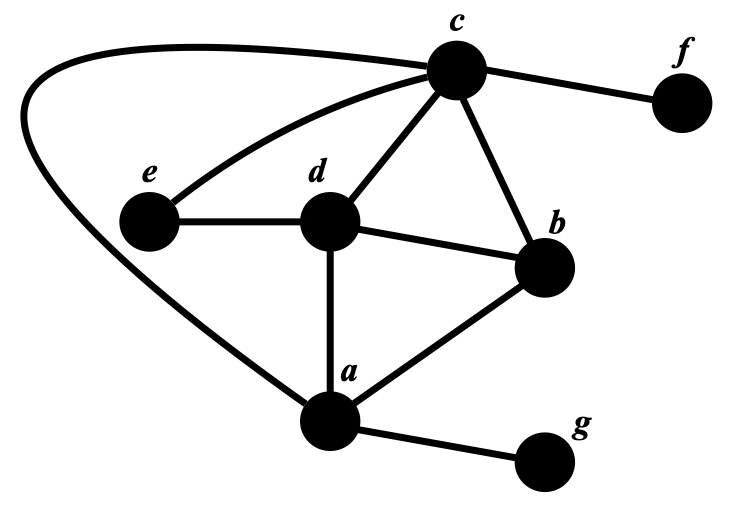}
\caption{Some part of a beeping network.}
\label{subfig:bn}
\end{subfigure}
\hspace{0.1in}
\begin{subfigure}[htbp]{0.30\textwidth}
\centering
\includegraphics[width=\linewidth]{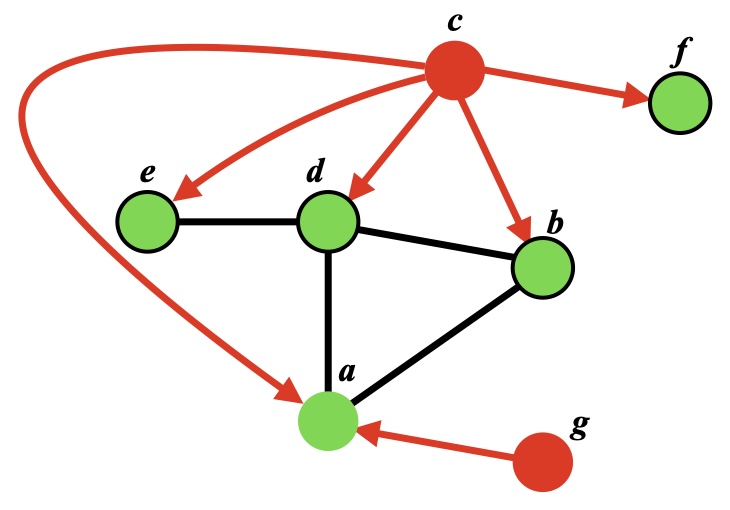}
\caption{Announcing super-round of some phase $j$: $\{c,g\}$ announce, $\{a,b,d,e,f\}$ hear noise, but only $\{b,d,e,f\}$ receive an extended-ID $\langle c\rangle$ and become $c$-$responsive$.}
\label{subfig:announce}
\end{subfigure}
\hspace{0.1in}
\begin{subfigure}[htbp]{0.30\textwidth}
\centering
\vspace*{-3ex}
\includegraphics[width=\linewidth]{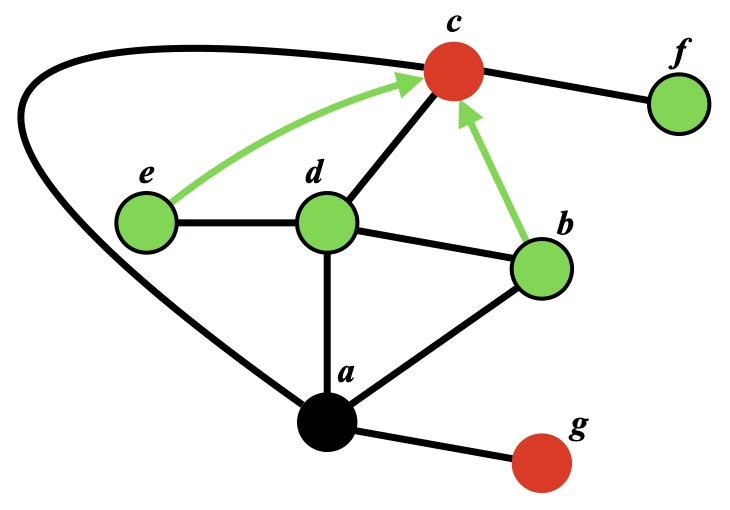}
\caption{Responding $3$ super-rounds within some sub-phase $a'$: $\{b,e\}$ respond and $c$ receives $\langle b\rangle\langle c\rangle\langle m_{b,c}\rangle$ and $\langle e\rangle\langle c\rangle\langle m_{e,c}\rangle$.}
\label{subfig:resp1}
\end{subfigure}
\\
\begin{subfigure}[htbp]{0.30\textwidth}
\centering
\vspace*{3ex}
\includegraphics[width=\linewidth]{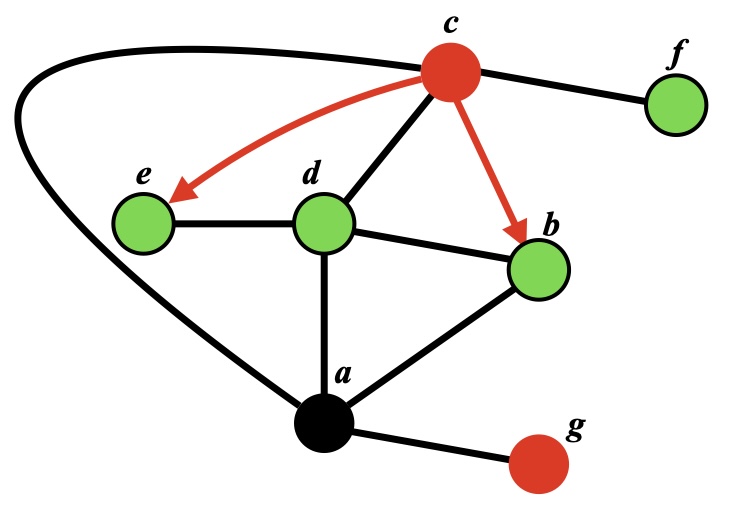}
\caption{Confirming $3$ super-rounds during sub-phase $a'$: $c$ confirms, $\{b,e\}$ receive $\langle c\rangle\langle b\rangle\langle m_{c,b}\rangle$ and $\langle c\rangle\langle e\rangle\langle m_{c,e}\rangle$ respectively. After this $\{b,e\}$ abandon the $c$-$responsive$ status and $\{\{b,c\},\{e,c\}\}$ are marked as realized.}
\label{subfig:conf1}
\end{subfigure}
\hspace{0.1in}
\begin{subfigure}[htbp]{0.30\textwidth}
\centering
\vspace*{-8ex}
\includegraphics[width=\linewidth]{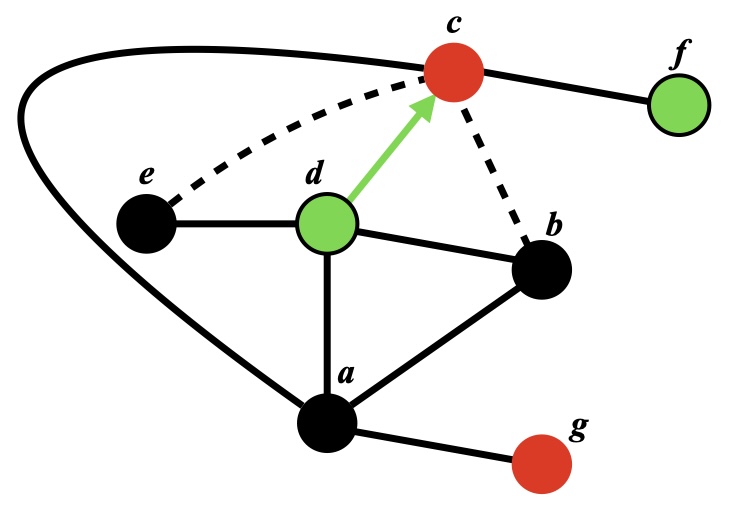}
\caption{Responding $3$ super-rounds within some sub-phase $a''$: $d$ responds and $c$ receives $\langle d\rangle\langle c\rangle\langle m_{d,c}\rangle$.}
\label{subfig:resp2}
\end{subfigure}
\hspace{0.1in}
\begin{subfigure}[htbp]{0.30\textwidth}
\centering
\vspace*{-3ex}
\includegraphics[width=\linewidth]{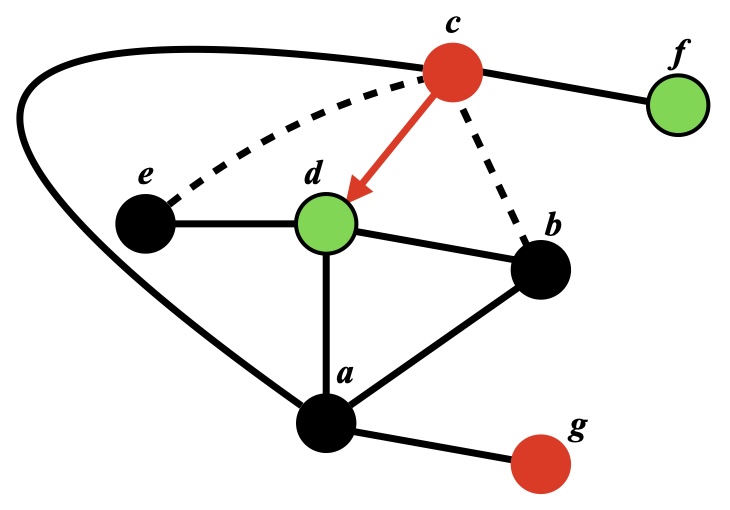}
\caption{Confirming $3$ super-rounds during sub-phase $a''$: $c$ confirms, $d$ receives $\langle c\rangle\langle d\rangle\langle m_{c,d}\rangle$. After this $d$ abandons the $c$-$responsive$ status and $\{d,c\}$ is marked as realized.}
\label{subfig:conf2}
\end{subfigure}
\\
\begin{subfigure}[htbp]{0.30\textwidth}
\centering
\includegraphics[width=\linewidth]{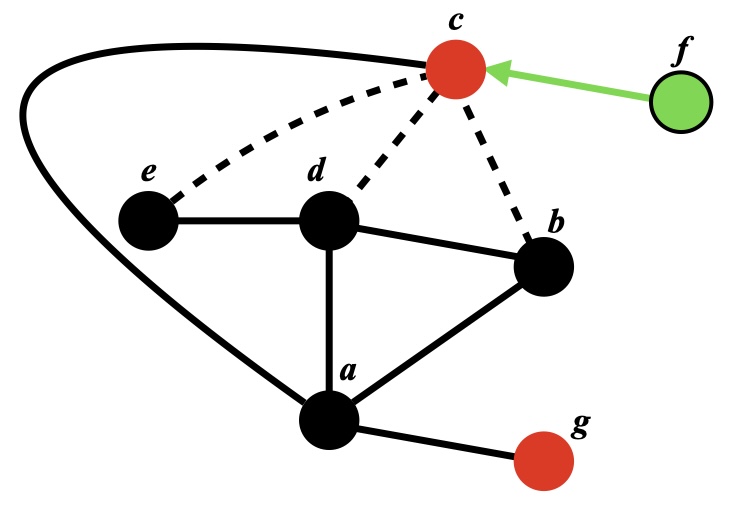}
\caption{Responding $3$ super-rounds within some sub-phase $a'''$: $f$ responds and $c$ receives $\langle f\rangle\langle c\rangle\langle m_{f,c}\rangle$.}
\label{subfig:resp3}
\end{subfigure}
\hspace{0.1in}
\begin{subfigure}[htbp]{0.30\textwidth}
\centering
\vspace*{5ex}
\includegraphics[width=\linewidth]{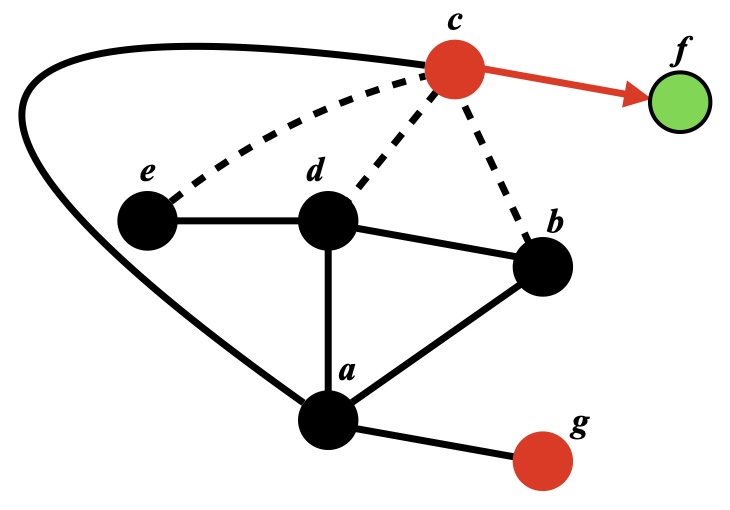}
\caption{Confirming $3$ super-rounds during sub-phase $a'''$: $c$ confirms, $f$ receives $\langle c\rangle\langle f\rangle\langle m_{c,f}\rangle$. After this $f$ abandons the $c$-$responsive$ status and $\{f,c\}$ is marked as realized.}
\label{subfig:conf3}
\end{subfigure}
\hspace{0.1in}
\begin{subfigure}[htbp]{0.30\textwidth}
\centering
\includegraphics[width=\linewidth]{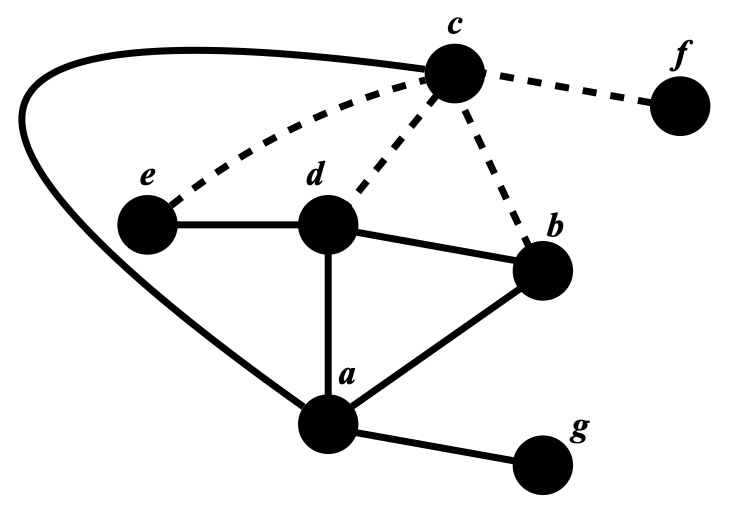}
\caption{By the end of phase $j$ links $\{\{b,c\},\{d,c\},\{e,c\},\{f,c\}\}$ have been realized.}
\label{subfig:end}
\end{subfigure}
\caption{Illustration of \alg algorithm -- consecutive handshakes between the announcer $c$ and its responders during a phase.}
\label{fig:alg}
\end{figure}

\begin{proof}[Proof of Lemma~\ref{lem:correct-receiving}]
The proof is by contradiction -- suppose that in some super-round a node $w$ 
%stays silent and 
receives an extended-ID $z$ but the claim of the lemma does not hold. 
Without lost of generality, we may assume that this is the first such super-round.

Recall that the definition of receiving an extended-ID requires that node $w$ has been silent in this super-round. Note that if exactly one neighbor of node $w$ has been beeping during the super-round, it must have been an extended-ID of some node (by specification of the algorithm), and therefore node $w$ receives this extended-ID (as other neighbors do not beep at all).
Similarly, we argue that at least one neighbor of node $w$ must have been beeping (some extended-ID) in the super-round, as otherwise node $w$ would not have received any beep (and so, also no extended-ID) in the considered super-round.

In the remainder we focus on the complementary case that at least two neighbors of node $w$ have been beeping in the super-round, each of them some extended-ID (again, by specification of the algorithm, a node beeps only some extended-ID or stay silent during any super-round).

First, suppose that some two neighbors, $v_1,v_2$, beeped different extended-IDs, say $z_1\ne z_2$, respectively.
It means that node $w$ received more than $\log n$ beeps during the super-round: $\log n$ beeps coming from one of the extended-IDs and at least one more because the extended IDs of different nodes differ by at least one position. Hence, the received sequence of beeps does not form any extended-ID, as it must always have $\log n$ bits $1$ corresponding to the beeps. This contradicts the fact that $w$ receives an extended-ID in the considered super-round.

Second, suppose that all extended-IDs beeped by (at least two) neighbors on node $w$ are the same. If this happens in an announcing, or a first responding, or a first confirming super-round, it is a contradiction because all nodes that beep in such super-rounds beep their own extended-IDs, which are pairwise different.
If this happens in a second confirming super-round, it means that these two neighbors $v_1,v_2$ belong to the same set $\mF_{\Delta,k_i}(j)$, for some phase number $j$, and both of them received an extended-ID of $z$ in the preceding responding super-round. By the fact that the considered super-round is the first when the lemma's claim does not hold, we get that in this preceding responding super-round, both $v_1,v_2$ received the extended-ID of $z$ when $z$ was their unique beeping neighbor (beeping its own extended-ID, by the specification of the responding super-rounds).
This, however, implies that in the beginning of the current phase $j$, i.e., during its announcing super-round, both $v_1,v_2$ beeped their extended-IDs and, again by the choice of the current contradictory super-round, their neighbor $w$ could not have received any extended-ID -- this is a contradiction with the fact that $z$ was transmitting in the responding super-round preceding the considered (contradictory) super-round. More precisely, only $(j,\cdot)$-responsive nodes can transmit in responding super-rounds, but $z$ is not $(j,\cdot)$-responsive because it had not received any extended-IDs in the first (announcing) super-round of the current phase.

The last sub-case of the above scenario, when all extended-IDs beeped by (at least two) neighbors on node $w$ are the same, is as follows. 
If this situation happens in a second responding super-round, it means 
that these two neighbors $v_1,v_2$ are both $(j,z)$-responsive and beep extended-ID of $z$. This is, however, acceptable due to the exception in the lemma's statement.

This completes the proof of the lemma.
\end{proof}

\begin{proof}[Proof of Lemma~\ref{lem:correct-realization}]
It is enough to show that points (a) and (b) of the definition of link realization occurred in the last four super-rounds (two responding and two confirming) and also that the other node, $w$, (locally) marks link $\{v,w\}$ as realized at the same time when $v$ does.

If node $v$ marked the link $\{v,w\}$ as realized, it could be because of one of two reasons. 

First, it is in the set $\mF_{\Delta,k_i}(j)$, where $i$ is the number of the current epoch and $j$ is the number of the current phase and received an extended-ID of $w$ followed by its own extended-ID in the preceding two responding super-rounds. This satisfies point (a) of the definition of link realization, as both were beeped by node $w$, by the algorithm specification, and by Lemma~\ref{lem:correct-receiving}. Note that the exception in that lemma does not really apply here because if there were two or more neighbors beeping the same extended-ID (of $v$) in the second responding round, they would also be beeping their own extended-IDs in the first responding round, which could contradict the fact that $v$ received a single extended-ID at that super round.

This also means that $v$ has beeped its own extended-ID followed by extended-ID of $w$ in the last two confirming super-rounds, which must have been received by $w$ because $w$ is $(j,v)$-responsive (because only such nodes could have beeped in the preceding responding super-rounds) and thus its only neighbor in set $\mF_{\Delta,k_i}(j)$ (only such nodes are allowed to beep in confirming super-rounds) is $v$; here we use Fact~\ref{fact:single-beeping}. Hence, $w$ also marks link 
$\{v,w\}$ as realized at the end of the two confirming super-rounds; by the algorithm's specification, point (c) of the definition also holds in this case.

Second, it is $(j,z)$-responsive and received an extended-ID of $z$ followed by its own extended-ID in the current two confirming super-rounds. This satisfies point (b) of the definition, as both were beeped by node $z$, by the specification of the algorithm and Lemma~\ref{lem:correct-receiving} (exception in that lemma does not apply here because we now consider only confirming super-rounds).

This also means that $v$ beeped its own extended-ID followed by extended-ID of $z$ in the preceding two responding super-rounds (because $v$ is $(j,z)$-responsive and only such nodes could beep in the preceding responding super-rounds), which must have been received by $z$ (otherwise, by the specification of the algorithm, $z$ would not beep its extended-ID followed by the extended-ID of $v$ in the last two confirming super-rounds). Hence, $z$ also marks link $\{v,z\}$ as realized at the end of the two confirming super-rounds by the algorithm's specification, and point (c) of the definition also holds in this case.
\end{proof}

\begin{proof}[Proof of Lemma~\ref{lem:subphase-progress}]
The lemma follows from the definition of the $(n,k_i/2^{a-2},k_i/2^{a-1})$-avoiding selector $\mF_{k_i/2^{a-2},k_i/2^{a-1}}$ used throughout sub-phase $a$ of phase $j$ of epoch $i$. 
By specification of the sub-phase, only nodes $w$ such that $w$ is $(j,v)$-responsive and it does not marked link $\{v,w\}$ as realized take active part in sub-phase $a$ (in the sense that only those nodes can beep extended-IDs of itself followed by $v$ in pairs of responding super-rounds), while other neighbors of $v$ do not beep at all. The latter statement needs more justification -- in the beginning of the current phase, in the announcing super-round, $v$ must have beeped because some nodes have become $(j,v)$-responsive in this phase (w.l.o.g. we may assume that at least one node has become $(j,v)$-responsive, because otherwise the lemma trivially holds), therefore, by Lemma~\ref{lem:correct-receiving}, other neighbors of $v$ could not receive another announcement and become $(j,v')$-responsive, for some $v'\ne v$, and thus by the description of the algorithm -- they stay silent throughout the whole phase. 

By lemma assumption, there are at most $\Delta/2^{i+a-2}=k_i/2^{a-2}$ $(j,v)$-responsive nodes $w$ that have not marked link $\{v,w\}$ as realized by the beginning of the sub-phase. Hence, at least half of them will be in a singleton intersection with some set $\mF_{k_i/2^{a-2},k_i/2^{a-1}}(b)$, by Definition~\ref{def:avoid-selector} and Fact~\ref{fact:avoiding-selectors}, in which case $v$ receives their beeping in the corresponding pair of the responding super-rounds. Consequently, $v$ beeps back its own extended-ID and the extended-ID of $w$ in the following two confirming super-rounds. 

Node $w$ receives those beepings, as there is no other neighbor of $w$ who is allowed to beep in these two rounds -- indeed if there was, it would belong to set $\mF_{\Delta,k_i}(j)$ and thus it would have been beeping in the announcing super-round of this phase, preventing (together with neighbor $v$ of $w$) node $w$ from receiving anything in that super-round (by Lemma~\ref{lem:correct-receiving}), which contradicts the fact that $w$ must have received an extended-ID of $v$ in that round to become $(j,v)$-responsive (as assumed). Therefore, by the description of the algorithm, $w$ marks link $\{v,w\}$ as realized. This completes the proof that the number of $(j,v)$-responsive neighbors $w$ of $v$ who remain without realizing link $\{v,w\}$ becomes less than $\Delta/2^{i+a-1}$ at the end of the considered sub-phase. 
\end{proof}

\begin{proof}[Proof of Lemma~\ref{lem:phase-progress}]
It follows directly from the fact that a phase, after its announcing super-round, iterates sub-phases $a=1,\ldots,\log k_i$. Each subsequent sub-phase halves the number of not-realized links $\{v,w\}$, for $(j,v)$-responsive nodes $w$ and each announcing node $v$, c.f., Lemma~\ref{lem:subphase-progress}, starting from the assumed $2k_i$ maximum number of $(j,v)$-responding nodes (recall that $(j,v)$-responding nodes form a subset of those to whom links are not realized, hence there are at most $2k_i$ of them in the beginning).
\end{proof}

\begin{proof}[Proof of Lemma~\ref{lem:epoch-invariant}]
The proof is by induction on epoch number $i$.
Obviously, the invariant holds at the beginning of the first epoch, i.e., $\kappa_1\le k_1$, where $\kappa_i$ was defined as the sharp upper bound on the maximum number of not realized links at a node at the end of epoch $i$ and $k_i$ is the parameter used in the algorithm for epoch $i$. 

Consider epoch $i\ge 1$.
We have to prove that:
assuming that $\kappa_{i'}\le k_{i'}$, for any $1\le i' < i$, we also have $\kappa_i \le k_i$.
Technically we can assume that $\kappa_0=k_0=\Delta$.

Consider a node $w$. 
By the inductive assumption, it has at most $k_{i-1}$ neighbors $v$ such that link $\{v,w\}$ has not been marked by $w$ as realized.
By Definition~\ref{def:avoid-selector} and Fact~\ref{fact:avoiding-selectors} applied to $(n,\Delta,\Delta-k_i)$-avoiding selector $\mF_{\Delta,k_i}$, 
which sets are used for announcing super-rounds (and later for confirming super-rounds), the number of neighbors $v$ of node $w$ from whom node $w$ has not received their extended-ID during the announcing super-round is smaller than $k_i$. By Lemma~\ref{lem:phase-progress}, all such nodes $w$ realize their links, and by Lemma~\ref{lem:correct-realization}, also node $v$ realizes these links during the considered phase. Hence, the number of non-realized links incident to any node $w$ drops below $k_i$ by the end of epoch $k_i$.
\end{proof}

%% file: primitives_details.tex
\section{Details of Section~\ref{sec:primitives} -- Algorithms and Analysis of Building Blocks}
\label{sec:prim_details}

In this section, we include the remaining details of our building blocks (Section~\ref{sec:primitives}). Theorems are restated for easy reference. First, let us introduce the following combinatorial object, to be used later.

\begin{definition}[Strong selector]
\label{def:strong-selector}
    A family $\mathcal{F}$ of subsets of $[n]$ of size at most $k$ each is called an \emph{$(n,k)$ strong selector} if for every non-empty subset $S \in [n]$ such that $|S| \leq k$, for every element $a \in S$, there exists a set $F \in \mathcal{F}$ such that $|F \cap S| = \{a\}$.
\end{definition}

Note that there are known constructions of $(n,\Delta)$ strong selectors of length at most $O(\Delta^2 \log n)$~\cite{5967914}. Next, we show how to use an $(n,\Delta)$ strong selector to perform a local broadcast.

\subsection{Local broadcast}
\label{sec:local-broadcast}

%One can perform local broadcast routine based on strong selectors. 
Our local broadcast routine is non-adaptive. That is, each node 
%will have a prepared-in-advance 
has a predefined schedule 
%determining 
specifying in which rounds beeps and in which rounds listens. 

\parhead{Assumptions} % \pga{The nodes IDs come from range $[1,n^c]$.} 
The nodes know the total number of nodes $n$, parameter $c$,
 %\mm{[[the range has been defined above, do we need to recall it?]]} 
the maximum degree $\Delta$ of the graph and have access to a global clock. Additionally, we assume that each node $v$ knows its neighborhood $N(v)$. (In Subsection~\ref{sub:neighbourhood} we will show how all nodes can learn their neighborhoods in $O(\Delta^2 \log^2 n)$ beeping rounds.)

\localbroadcastthm*
\remove{
\begin{theorem}
\label{th:local_broadcast}
\mm{Consider a Beeping Network where each node $v$ knows $n$, $\Delta$, and $N(v)$.}
    Assume that every message $m_v$ of each node $v$ has length at most $k$ for some $k>0$. Then, there is a deterministic distributed local broadcast algorithm that works in $O(k\Delta^2 \log^2 n)$ beeping rounds.
\end{theorem}
}

\begin{proof}
Consider an $(n^c,\Delta)$ strong selector $\mathcal{F}=\{S_1,S_2,\dots, S_L\}$ of length $L=O(\Delta^2 \log n)$, known to all nodes. Our local broadcast schedule will take $L$ rounds. At any round $i$, nodes $v \in S_i$ that have bit 1 (indicating to transmit) send a beep while all the other nodes are silent.

Consider any receiver $r$. Consider the set $N(r)$ of neighbors of $r$. Note that $|N(r)| \leq \Delta$. From the definition of a $(n^c,\Delta)$ strong selector $\mathcal{F}$, for every $v \in N(r)$ there exists an index $i$ such that $S_i \cap N(r) = \{v\}$. Therefore, for every pair of transmitters $v$ and receiver $r$ that are adjacent to each other, there exists a round $i$ such that $v$ is the only transmitting neighbor of $r$.

%Here we will assume that each node knows its neighborhood. In Subsection~\ref{sub:neighbourhood} we will show how all the nodes can learn their neighborhoods in $O(\Delta^2 \log^2 n)$ beeping rounds.

Since every node $v$ knows its neighborhood $N(v)$ and the sets $S_i$ for all $i$, node $v$ also knows for each neighbor $u \in N(v)$ at what round $t$ neighbor $u$ is the only neighbor transmitting. If at round $t$ node $v$ hears a beep, it means that $u$ transmitted bit 1. If at round $t$ node $v$ hears silence, it means that $u$ transmitted bit 0.

Therefore, after $L$ rounds the algorithm will go through the entire strong selector $\mathcal{F}$ and each node will learn a bit of information from each of its neighbors.

The procedure can be repeated $B$ times to broadcast messages of at most $B$ bits.
Hence, the claim follows.
\end{proof}

\subsection{Learning neighborhood}
\label{sub:neighbourhood}

Now we show how all the nodes can learn their neighborhoods in $O(\Delta^2 \log^2 n)$ beeping rounds. The following procedure will be non-adaptive. 

\parhead{Assumptions} %\pga{The nodes IDs come from range $[1,n^c]$.} 
The nodes know the total number of nodes $n$ and parameter $c$.
% \pga{the range of possible IDs $[1,n^c]$ for some constant $c \geq 1$}
% and have access to a global clock.

\learningneighthm*
\remove{
\begin{theorem}
\label{th:learning_neighbourhood}
\mm{Consider a Beeping Network where each node $v$ knows $n$.}
    There is a deterministic distributed learning neighborhood algorithm that works in $O(\Delta^2 \log^2 n)$ beeping rounds.
\end{theorem}
}

\begin{proof}
The nodes will beep according to a strong selector $\mF$ as in the previous subsection. This time, however, for each set $S_i \in \mF$ there will be $2\log n$ beeping rounds instead of $1$ round. First, nodes $v \in S_1$ will transmit for $2\log n$ rounds, then nodes $v \in S_2$ will transmit for $2\log n$ rounds and so on.

In each block of $2\log n$ rounds corresponding to a set $S_i$ for some $i$, each node $v \in S_i$ will encode its ID in the following way. For each bit in its ID, if the bit is 1, then the node listens for 1 round and then beeps. If the bit is 0, then the node beeps once and then listens for 1 round. The process is repeated for each bit in the ID.

After all $2\log n$ beeping rounds corresponding to a set $S_i$ for some $i$ pass, each node $v$ can look at the string of beeps and silences that it heard during the block. If there were beeps in both rounds $2k$ and $2k+1$ for some $k$, then there were multiple neighbors transmitting in the block and $v$ will ignore this block. Otherwise, the string of beeps encodes the ID of the only transmitting neighbor $u$. 
% node $v$ will assume that this string of bits is an ID of a node $u$. If the ID of $u$ is such that $u \in S_i$, then $u$ was the only neighbor of $v$ transmitting in the block. 
% \pga{[We may need to transmit 01 for each bit 1 and 10 for each bit 0 to really make this unambiguous.]} 
In particular, $u$ can be added to the list of neighbors of $v$. 
% On the other hand, if $u \notin S_i$, then that means that multiple neighbors of $v$ were overlapping their beeps during the block corresponding to $S_i$ and $u$ should not be added to the list of neighbors of $v$.

% \pga{According to the definition of strong selector $\mF$,} 
After all $L$ blocks of transmissions pass, each node $v$ heard from each $u \in N(v)$ in a block such that $u$ was the only transmitting neighbor. Therefore, each $u \in N(v)$ was added to the list of neighbors of $v$. No other nodes were added to the list of neighbors of $v$. Thus, $v$ knows its neighborhood from now on and the claim follows.
\end{proof}

\subsection{Cluster gathering}
\label{sec:gathering}

% \noindent \textbf{Aggregating information via overlapping Steiner trees.} 

\noindent \textbf{Assumptions.} %\pga{The nodes IDs come from range $[1,n^c]$.} 
Thanks to running cluster gathering inside the network decomposition algorithm, we will have access to additional structures. During the working of the network decomposition algorithm, each cluster $C$ will have a Steiner tree $S$ associated with it. All nodes $v \in C$ will be regular nodes in the Steiner tree $S$, while there may be some additional nodes $u \notin C$ that are Steiner nodes in $S$. All Steiner trees will have depth at most $O(\log^2 n)$, i.e., the diameter of the Steiner tree $S$ will be the same as the weak-diameter of the cluster $C$ that corresponds to $S$. Each node and each edge will be in at most $O(\log n)$ Steiner trees. Each Steiner tree $S$ will have a fixed root node $r$.
Given that our cluster gathering algorithm uses the local broadcast algorithm defined above, the previous assumptions also~apply.

% We can develop a cluster gathering algorithm where all nodes in all clusters work in parallel, resulting in only $O(local\_broadcast\cdot \log^4 n)$ beeping rounds.

\noindent \textbf{Effect and efficiency.} Given the above assumptions, we can develop an algorithm that gathers and aggregates limited information (each node passes at most $O(\log n)$ bits for each Steiner tree it is in) from all nodes in a cluster $C$ to the root $r$ of the corresponding Steiner tree $S$, with all clusters working in parallel. Additionally, the root $r$ can broadcast $O(\log n)$ bits to all nodes in $C$. The algorithm will work in $O(\Delta^2 \cdot \log^6 n)$ beeping rounds.

\noindent \textbf{Utilization.} The gathering algorithm will be used throughout the network decomposition algorithm but may be of independent use, especially since the network decomposition algorithm may output the Steiner trees it was using as well as the decomposition. Therefore, any algorithm using our network decomposition algorithm will have access to the appropriate Steiner trees to use for %gathering 
collecting information using our cluster gathering algorithm.

\clustergatherthm*
\remove{
\begin{theorem}
\label{th:cluster_gathering}
\mm{Consider a Beeping Network where each node $v$ knows $n$, $\Delta$, and $N(v)$.}
    There is a deterministic distributed cluster gathering algorithm that works in $O(\Delta^2 \log^4 n)$ beeping rounds.
\end{theorem}
}

%\noindent \textbf{Algorithm description.} 
\begin{proof}
The algorithm will utilize the local broadcast subroutine at each step. When we say "transmit", we mean to use local broadcast subroutine, unless specified otherwise. Each node $v \in C$ broadcasts its $O(\log n)$ bits (e.g., a number less than $n$) in parallel using the local broadcast subroutine $O(\log n)$ times. At each step, the parent $p$ of $v$ in the corresponding Steiner tree will listen for the message from $v$ as well as messages from its other children. Whenever $p$ receive messages from some of its children, $p$ prepares its own message (e.g., sum of numbers provided by its children in the current step, assuming that the sum is smaller than $n$) of $O(\log n)$ bits. The process will repeat until the root $r$ receives all the messages, which will last the number of steps equal to the depth of the Steiner tree, $O(\log^2 n)$. Note that node $p$ may receive messages from its children in multiple steps; in that case in each step $t$ node $p$ transmits the aggregate of messages it received at $t$, thus transmitting multiple times.

Similarly, one can broadcast a message from $r$ to all the nodes $v\in C$ in $O(\log^2 n)$ steps. The entire algorithm takes $O(\log^2 n)$ steps, which is $O(t_{local\_broadcast} \cdot \log^2 n)$ beeping rounds and the claim follows.
\end{proof}

%%%%%%%%%%%%%%%%%%%%%%%%%%%%%%%%%%%%%%%%%%%%%%%%%%%%

\subsection{Network Decomposition}
\label{sec:decomposition}

% \subsection{Network decomposition from~\cite{ghaffari2021improved}}

In this section, we present how to adapt the network decomposition algorithm of Ghaffari et al.~\cite{ghaffari2021improved} to the beeping model. The only changes are in the way that nodes communicate. The original algorithm was %made
designed for the \congest model, where communication was straightforward. 
Instead, for the beeping model, we have to carefully implement all the concurrent communication, so that the algorithm remains efficient.
First, let us recall the result from~\cite{ghaffari2021improved}. 

% \noindent\textbf{Notations:} $b$ is the length of identifiers, $n$ is the number of nodes in graph $G$.

\begin{theorem}\cite{ghaffari2021improved}
    There is a deterministic distributed algorithm that computes a $(\log n, \log^2 n)$ network decomposition in $O(\log^5 n)$ \congest rounds.
\end{theorem}

We adapt the above result to the beeping model and obtain the next theorem.

\networkdecompthm*
\remove{
\begin{theorem}
\label{thm:local-decomposition}
    There is a deterministic distributed algorithm that computes a $(\log n, \log^2 n)$ network decomposition in $O(\Delta^2 \log^8 n)$ beeping rounds.
\end{theorem}
}

The algorithm works as follows. In preprocessing, each node learns its neighborhood, so that we will be able to use local broadcast routine.
%freely. 
This part can take up to $O(\Delta^2 \log^2 n)$ beeping rounds (see Theorem~\ref{th:learning_neighbourhood}). Next, we perform the network decomposition algorithm from~\cite{ghaffari2021improved} with all the communication carefully replaced, as shown 
%in Section~\ref{sec:rundown} 
below. This completes our network decomposition algorithm.

% \dk{In order to combine our Local Broadcast with the tools in~\cite{ghaffari2021improved}, detailed check of these tools need to be done to assure, among others, that they rely on Local Broadcast (not the general \congest model rules).
% We provide the relevant parts in Section~\ref{sec:Ghaffari-verbatim}.}

%\subsubsection{Summary of messages}
\label{sec:rundown}

% Let $local\_broadcast$ denote the number of beeping rounds to perform a local broadcast in the beeping model, i.e., make sure that for every node $v$ all neighbors of $v$ receive a message from $v$ \pga{of length at most $\log n$. According to Theorem~\ref{th:local_broadcast}, such local broadcast lasts $O(\Delta^2 \log^2 n$ beeping rounds.}

\paragraph{Summary of messages:}
We need to carefully replace all the communication from~\cite{ghaffari2021improved} with routines that work in the beeping model. We will use our local broadcast and cluster gathering routines from Section~\ref{sec:primitives}.

% Below we present a list of messages transmitted in \pga{the original algorithm for the \congest model~\cite{ghaffari2021improved} and how to implement this information transmission in the beeping model. The original algorithm can be viewed in Subsection~\ref{sec:ghaffari}. In the list below we count the messages done per \emph{step}. There will be at most $O(\log^2 n)$ steps in the algorithm.}

Below, we present a list of messages transmitted in %Subsection~\ref{sec:Ghaffari-verbatim} in 
a step and how to implement this information %transmission 
propagation in the beeping model (for convenience, we attach the network decomposition algorithm from~\cite{ghaffari2021improved} in Section~\ref{sec:Ghaffari-verbatim}):
\begin{enumerate}
    \item \textbf{Before proposals.} A node $v$ needs to check which clusters are adjacent to it and what are the parameters of these clusters (id, level). Every node $u$ can broadcast the parameters of the cluster $u$ is in, which is at most $O(\log n)$ bits. This is done once per step, %. This used to cost
    at a cost of $O(1)$ \congest rounds per step. In the beeping model, it can be done via a local broadcast with messages of length at most $O(\log n)$ bits. According to Theorem~\ref{th:local_broadcast}, it will cost $O(\Delta^2 \log^2 n)$ beeping rounds per step.
    \item \textbf{Proposals.} Proposals to join a cluster can be transmitted to all neighbors if the target cluster (or the target node) is specified in the message. Other receivers may ignore the message. This part used to cost $O(1)$ \congest rounds per step. 
    % In beeping model, this can be done in $O(local\_broadcast)$ rounds per step.
    In the beeping model, it can be done via a local broadcast, where the message specifies the IDs of both, sender and receiver, meaning messages of length at most $O(\log n)$ bits. According to Theorem~\ref{th:local_broadcast}, it will cost $O(\Delta^2 \log^2 n)$ beeping rounds per step.
    \item \textbf{Gathering proposals.} The leader of each cluster should learn the total number of proposals and the total number of tokens in the cluster. This part was done in $O(\log^3 n)$ \congest rounds. %of \congest model. 
    Note that the algorithm keeps the Steiner tree of $O(\log^2 n)$ diameter for each cluster, such that each node (and therefore each edge) is in at most $O(\log n)$ Steiner trees. Additionally, each node that participates in this cluster gathering  can add the numbers of proposals/tokens it receives from other nodes and then transmit the sums instead of relaying each message separately; these sums can never exceed the total number of nodes $n$, so this guarantees that each node only transmits $O(\log n)$ bits per Steiner tree it is in. Therefore, we can use the cluster gathering algorithm. According to Theorem~\ref{th:cluster_gathering}, this takes at most $O(\Delta^2 \log^6 n)$ beeping rounds.
    % given a $O(\log^2 n)$ diameter Steiner tree such that each node (and therefore each edge) is in at most $O(\log n)$ Steiner trees. 
    % This can be done in $O(\log^3 n)$ rounds of \congest model, since a message from a leaf to the root can collide with at most $O(\log^3 n)$ different messages. In beeping model, we can use the 
    % this can be done in $O(local\_broadcast \cdot \log^3 n)$ rounds per step. \pga{There should be an additional factor of $O(\log n)$ due to the length of the messages passed. Also, how do we determine, which message (from which Steiner tree) to transmit first?}
    \item \textbf{Responding to proposals.} Each node informs all its neighbors either that all the proposals were accepted or that all the proposing nodes should be killed. This part used to cost $O(1)$ \congest rounds per step. In the beeping model, the same can be done in a local broadcast, with $O(1)$ bit messages. According to Theorem~\ref{th:local_broadcast}, it will cost $O(\Delta^2 \log n)$ beeping rounds.
    \item \textbf{Stalling.} If a cluster decides to stall, all nodes neighboring the cluster should be informed about it. This \mm{part} used to cost $O(1)$ \congest rounds per step. In the beeping model, \mm{the same} can be done in a local broadcast with $O(1)$ bit messages. According to Theorem~\ref{th:local_broadcast}, it will cost $O(\Delta^2 \log n)$ beeping rounds.
\end{enumerate}

In the procedure described above, there is a total of $O(\Delta^2 \cdot \log^6 n)$ beeping rounds required per step. There are up to $O(\log n)$ steps per phase and up to $O(\log n)$ phases in the algorithm, which results in $O(\Delta^2 \cdot \log^8 n)$ beeping rounds for the entire algorithm. Note that only the means of communication changed. Therefore, the correctness of the algorithm is unaffected. This completes the proof of Theorem~\ref{thm:local-decomposition}.

\subsection{GGR network decomposition algorithm}
\label{sec:Ghaffari-verbatim}

\noindent\textbf{Notation:} $b$ is the length of identifiers, $n$ is the number of nodes in graph $G$.

The remainder of this subsection, which is important from perspective of assurance that our Local Broadcast could be combined with the tools in~\cite{ghaffari2021improved}, is cited from~\cite{ghaffari2021improved} verbatim.

\noindent\textbf{Construction  outline:}   The  construction  has  $2(b+\log n) =O(\log n)$ phases. Each phase has $28(b+\log n) =O(\log n)$ steps.  Initially, all nodes of $G$ are \emph{living}, during the construction some living nodes \emph{die}. Each living node is  part  of  exactly  one  cluster.   Initially,  there  is  one cluster $C_v$ for each vertex $v\in V(G)$ and we define the identifier $id(C)$ of $C$ as the unique identifier of $v$ and use $id_i(C)$  to  denote  the $i$-th  least  significant  bit  of  $id(C)$. From now on, we talk only about identifiers of clusters and do not think of vertices as having identifiers, though they will still use them for simple symmetry breaking tasks.   Also,  at  the  beginning,  the  Steiner  tree $T_{C_v}$ of a cluster $C_v$ contains just one node, namely $v$ itself, as a  terminal  node.   Clusters  will  grow  or  shrink  during the iterations, while their Steiner trees collecting their vertices can only grow.  When a cluster does not contain any nodes, it does not participate in the algorithm anymore.

\noindent\textbf{Parameters of each cluster:} Each cluster $C$ keeps two other parameters besides its identifier $id(C)$ to make its decisions:  its number of tokens $t(C)$ and its level $lev(C)$.The number of tokens can change in each step -- more precisely it is incremented by one whenever a new vertex joins $C$, while it does not decrease when a vertex leaves $C$.  The number of tokens only decreases when $C$ actively deletes nodes.  We define $t_i(C)$ as the number of tokens of $C$ at the beginning of the $i$-th phase and set $t_1(C) = 1$. Each  cluster  starts  in  level  $0$.   The  level  of  each cluster does not change within a phase $i$ and can only increment by one between two phases; it is bounded by $b$.  We denote with $lev_i(C)$ the level of $C$ during phase $i$.   Moreover,  for  the  purpose  of  the  analysis,  we  keep track  of  the  potential  $\Phi(C)$  of  a  cluster $C$ defined  as $\Phi_i(C) = 3i - 2lev_i(C) + id_{lev_i(C)+1}(C)$.  The potential of each cluster stays the same within a phase.

\noindent\textbf{Description  of  a  step:} In each step, first, each node $v$ of each cluster $C$ checks whether it is adjacent to a  cluster $C'$ such that  $lev(C')<lev(C)$. If  so, then $v$ proposes  to  an arbitrary  neighboring  cluster $C'$ among the neighbors with the smallest level $lev(C')$ and if there is a choice, it prefers to join clusters with $id_{lev(C')+1}(C') = 1$.  Otherwise, if there is a neighboring cluster $C'$ with $lev(C') = lev(C)$ and $id_{lev(C')+1}(C') = 1$, while  $id_{lev(C)+1}(C)  =  0$,  then $v$ proposes  to  arbitrary such cluster.

Second, each cluster $C$ collects the number of proposals  it  received.   Once  the  cluster  has  collected  the number  of  proposals,  it  does  the  following.   If  there are $p$ proposing nodes,  then they join $C$ if and only if $p \geq t(C)/(28(b+ \log n))$.  The denominator is equal to the number of steps. If $C$ accepts these proposals, then $C$ receives $p$ new tokens, one from each newly joined node. On the other hand, if $C$ does not accept the proposals as their number is not sufficiently large, then $C$ decides to kill all those proposing nodes.  These nodes are then removed from $G$.  Cluster $C$ pays $p \cdot 14(b+ \log n)$ tokens for this, i.e., it pays $14(b+ \log n)$ tokens for every vertex that it deletes.  These tokens are forever gone.  Then the cluster does not participate in growing anymore,  until the end of the phase and throughout that time we call that cluster \emph{stalling}.  The cluster tells that it is stalling to neighboring nodes so that they do not propose to it. At the end of the phase, each stalling cluster increments its level by one.

If the cluster is in level $b-1$ and goes to the last level $b$, it will not grow anymore during the whole algorithm, and  we  say  that  it  has finished.    Other  neighboring clusters can still eat its vertices (by this we mean that vertices of the finished clusters may still propose to join other clusters). 

Whenever  a  node $u$ joins  a  cluster $C$ via  a  vertex $v\in C$, we add $u$ to the Steiner tree $T_C$ as a new terminal node and connect it via an edge $uv$.  Whenever a node $u\in C$ is deleted or eaten by a different cluster, it stays in the Steiner tree $T_C$ but is changed to a non-terminal node.

%% file: conclude.tex
\section{Conclusions}

We provided deterministic distributed algorithms to efficiently simulate a round of algorithms designed for the CONGEST model on the Beeping Networks. This allowed us to improve polynomially the time complexity of several (also graph) problems on Beeping  Networks. The first simulation by the Local Broadcast algorithm is shorter by a polylogarithmic factor than the other, more general one -- yet still powerful enough to implement some algorithms, including the prominent solution to Network Decomposition~\cite{ghaffari2021improved}.
The more general one could be used for solving problems such as MIS.
We also considered efficient pipelining of messages via several layers of BN.
%We also proved that our solutions could not be substantially improved if the considered problems require content-oblivious local broadcast, by proving an almost-tight lower bound.

Two important lines of research arise from our work.
First, whether some (graph) problems do not need local broadcast to be solved deterministically, and whether their time complexity could be asymptotically below $\Delta^2$.
Second, could a lower bound on any deterministic local broadcast algorithm, better than $\Omega(\Delta\log n)$, be proved?
%our lower bound be tightened and extended to any, not necessarily content-oblivious \mam{and non-adaptive}, solutions to the Local Broadcast problem?

% \todo{Propose to develop algorithms that work in time depending on the diameter of the network}

% \todo{Discussion of noisy beeping channel.}

%% file: main.bbl
\begin{thebibliography}{10}

\bibitem{afek2013beeping}
Yehuda Afek, Noga Alon, Ziv Bar-Joseph, Alejandro Cornejo, Bernhard Haeupler,
  and Fabian Kuhn.
\newblock Beeping a maximal independent set.
\newblock {\em Distributed computing}, 26(4):195--208, 2013.

\bibitem{afek2011biological}
Yehuda Afek, Noga Alon, Omer Barad, Eran Hornstein, Naama Barkai, and Ziv
  Bar-Joseph.
\newblock A biological solution to a fundamental distributed computing problem.
\newblock {\em science}, 331(6014):183--185, 2011.

\bibitem{ashkenazi2020brief}
Yagel Ashkenazi, Ran Gelles, and Amir Leshem.
\newblock Brief announcement: Noisy beeping networks.
\newblock In {\em Proceedings of the 39th Symposium on Principles of
  Distributed Computing}, pages 458--460, 2020.

\bibitem{beauquier2019optimal}
Joffroy Beauquier, Janna Burman, Peter Davies, and Fabien Dufoulon.
\newblock Optimal multi-broadcast with beeps using group testing.
\newblock In {\em Structural Information and Communication Complexity: 26th
  International Colloquium, SIROCCO 2019, L'Aquila, Italy, July 1--4, 2019,
  Proceedings 26}, pages 66--80. Springer, 2019.

\bibitem{beauquier2018fast}
Joffroy Beauquier, Janna Burman, Fabien Dufoulon, and Shay Kutten.
\newblock Fast beeping protocols for deterministic mis and ($\delta$+
  1)-coloring in sparse graphs.
\newblock In {\em IEEE INFOCOM 2018-IEEE Conference on Computer
  Communications}, pages 1754--1762. IEEE, 2018.

\bibitem{BonisGV05}
Annalisa~De Bonis, Leszek Gasieniec, and Ugo Vaccaro.
\newblock Optimal two-stage algorithms for group testing problems.
\newblock {\em {SIAM} J. Comput.}, 34(5):1253--1270, 2005.

\bibitem{casteigts2019design}
Arnaud Casteigts, Yves M{\'e}tivier, John~Michael Robson, and Akka Zemmari.
\newblock Design patterns in beeping algorithms: Examples, emulation, and
  analysis.
\newblock {\em Information and Computation}, 264:32--51, 2019.

\bibitem{chlebus2017naming}
Bogdan~S Chlebus, Gianluca De~Marco, and Muhammed Talo.
\newblock Naming a channel with beeps.
\newblock {\em Fundamenta Informaticae}, 153(3):199--219, 2017.

\bibitem{ChlebusK05}
Bogdan~S. Chlebus and Dariusz~R. Kowalski.
\newblock Almost optimal explicit selectors.
\newblock In Maciej Liskiewicz and R{\"{u}}diger Reischuk, editors, {\em
  Fundamentals of Computation Theory, 15th International Symposium, {FCT} 2005,
  L{\"{u}}beck, Germany, August 17-20, 2005, Proceedings}, volume 3623 of {\em
  Lecture Notes in Computer Science}, pages 270--280. Springer, 2005.

\bibitem{CLEMENTI2003337}
Andrea~E.F. Clementi, Angelo Monti, and Riccardo Silvestri.
\newblock Distributed broadcast in radio networks of unknown topology.
\newblock {\em Theoretical Computer Science}, 302(1):337--364, 2003.

\bibitem{cornejo2010deploying}
Alejandro Cornejo and Fabian Kuhn.
\newblock Deploying wireless networks with beeps.
\newblock In {\em Distributed Computing: 24th International Symposium, DISC
  2010, Cambridge, MA, USA, September 13-15, 2010. Proceedings 24}, pages
  148--162. Springer, 2010.

\bibitem{czumaj2019communicating}
Artur Czumaj and Peter Davies.
\newblock Communicating with beeps.
\newblock {\em Journal of Parallel and Distributed Computing}, 130:98--109,
  2019.

\bibitem{davies2023optimal}
Peter Davies.
\newblock Optimal message-passing with noisy beeps.
\newblock In {\em Proceedings of the 2023 ACM Symposium on Principles of
  Distributed Computing}, pages 300--309, 2023.

\bibitem{degesys2007desync}
Julius Degesys, Ian Rose, Ankit Patel, and Radhika Nagpal.
\newblock Desync: Self-organizing desynchronization and tdma on wireless sensor
  networks.
\newblock In {\em Proceedings of the 6th international conference on
  Information processing in sensor networks}, pages 11--20, 2007.

\bibitem{dufoulon2018beeping}
Fabien Dufoulon, Janna Burman, and Joffroy Beauquier.
\newblock Beeping a deterministic time-optimal leader election.
\newblock In {\em 32nd International Symposium on Distributed Computing (DISC
  2018)}. Schloss Dagstuhl-Leibniz-Zentrum fuer Informatik, 2018.

\bibitem{dufoulon2022beeping}
Fabien Dufoulon, Yuval Emek, and Ran Gelles.
\newblock Beeping shortest paths via hypergraph bipartite decomposition.
\newblock {\em arXiv preprint arXiv:2210.06882}, 2022.

\bibitem{efremenko2018interactive}
Klim Efremenko, Gillat Kol, and Raghuvansh Saxena.
\newblock Interactive coding over the noisy broadcast channel.
\newblock In {\em Proceedings of the 50th Annual ACM SIGACT Symposium on Theory
  of Computing}, pages 507--520, 2018.

\bibitem{flury2010slotted}
Roland Flury and Roger Wattenhofer.
\newblock Slotted programming for sensor networks.
\newblock In {\em Proceedings of the 9th ACM/IEEE International Conference on
  Information Processing in Sensor Networks}, pages 24--34, 2010.

\bibitem{forster2014deterministic}
Klaus-Tycho F{\"o}rster, Jochen Seidel, and Roger Wattenhofer.
\newblock Deterministic leader election in multi-hop beeping networks.
\newblock In {\em Distributed Computing: 28th International Symposium, DISC
  2014, Austin, TX, USA, October 12-15, 2014. Proceedings 28}, pages 212--226.
  Springer, 2014.

\bibitem{ghaffari2021improved}
Mohsen Ghaffari, Christoph Grunau, and Vaclav Rozhoň.
\newblock Improved deterministic network decomposition.
\newblock In {\em Proceedings of the 2021 ACM-SIAM Symposium on Discrete
  Algorithms (SODA)}, pages 2904--2923. SIAM, 2021.

\bibitem{ghaffari2013near}
Mohsen Ghaffari and Bernhard Haeupler.
\newblock Near optimal leader election in multi-hop radio networks.
\newblock In {\em Proceedings of the twenty-fourth annual ACM-SIAM Symposium on
  Discrete algorithms (SODA)}, pages 748--766. SIAM, 2013.

\bibitem{haas2002dual}
Zygmunt~J Haas and Jing Deng.
\newblock Dual busy tone multiple access (dbtma)-a multiple access control
  scheme for ad hoc networks.
\newblock {\em IEEE transactions on communications}, 50(6):975--985, 2002.

\bibitem{holzer2016brief}
Stephan Holzer and Nancy Lynch.
\newblock Brief announcement: beeping a maximal independent set fast.
\newblock In {\em 30th International Symposium on Distributed Computing
  (DISC)}, 2016.

\bibitem{hounkanli2020global}
Kokouvi Hounkanli, Avery Miller, and Andrzej Pelc.
\newblock Global synchronization and consensus using beeps in a fault-prone
  multiple access channel.
\newblock {\em Theoretical Computer Science}, 806:567--576, 2020.

\bibitem{hounkanli2015deterministic}
Kokouvi Hounkanli and Andrzej Pelc.
\newblock Deterministic broadcasting and gossiping with beeps.
\newblock {\em arXiv preprint arXiv:1508.06460}, 2015.

\bibitem{hounkanli2016asynchronous}
Kokouvi Hounkanli and Andrzej Pelc.
\newblock Asynchronous broadcasting with bivalent beeps.
\newblock In {\em Structural Information and Communication Complexity: 23rd
  International Colloquium, SIROCCO 2016, Helsinki, Finland, July 19-21, 2016,
  Revised Selected Papers 23}, pages 291--306. Springer, 2016.

\bibitem{jeavons2016feedback}
Peter Jeavons, Alex Scott, and Lei Xu.
\newblock Feedback from nature: simple randomised distributed algorithms for
  maximal independent set selection and greedy colouring.
\newblock {\em Distributed Computing}, 29:377--393, 2016.

\bibitem{MooreNeuron24}
Jason~J. Moore, Alexander Genkin, Magnus Tournoy, Joshua~L. Pughe-Sanford,
  Rob~R. de~Ruyter~van Steveninck, and Dmitri~B. Chklovskii.
\newblock The neuron as a direct data-driven controller.
\newblock {\em Proceedings of the National Academy of Sciences},
  121(27):e2311893121, 2024.

\bibitem{motskin2009lightweight}
Arik Motskin, Tim Roughgarden, Primoz Skraba, and Leonidas Guibas.
\newblock Lightweight coloring and desynchronization for networks.
\newblock In {\em IEEE INFOCOM 2009}, pages 2383--2391. IEEE, 2009.

\bibitem{navlakha2014distributed}
Saket Navlakha and Ziv Bar-Joseph.
\newblock Distributed information processing in biological and computational
  systems.
\newblock {\em Communications of the ACM}, 58(1):94--102, 2014.

\bibitem{peleg2000distributed}
David Peleg.
\newblock {\em Distributed computing: a locality-sensitive approach}.
\newblock SIAM, Philadelphia, PA, USA, 2000.

\bibitem{5967914}
Ely Porat and Amir Rothschild.
\newblock Explicit nonadaptive combinatorial group testing schemes.
\newblock {\em IEEE Transactions on Information Theory}, 57(12):7982--7989,
  2011.

\bibitem{tobagi1975packet}
Fouad Tobagi and Leonard Kleinrock.
\newblock Packet switching in radio channels: Part ii-the hidden terminal
  problem in carrier sense multiple-access and the busy-tone solution.
\newblock {\em IEEE Transactions on communications}, 23(12):1417--1433, 1975.

\end{thebibliography}
